\documentclass[letterpaper]{article} 
\usepackage{aaai24}  
\usepackage{times}  
\usepackage{helvet}  
\usepackage{courier}  
\usepackage[hyphens]{url}  
\usepackage{graphicx} 
\usepackage{natbib}  
\usepackage{caption} 
\usepackage{algorithm}
\usepackage{amsmath}
\usepackage{bm}
\usepackage{amsthm}
\newtheorem{theorem}{Theorem}

\newtheorem{corollary}{Corollary}
\newtheorem{definition}{Definition}

\usepackage{bibentry}
\usepackage{booktabs}

\usepackage{xcolor}

\usepackage{newfloat}
\usepackage{listings}
\DeclareCaptionStyle{ruled}{labelfont=normalfont,labelsep=colon,strut=off} 
\floatstyle{ruled}
\newfloat{listing}{tb}{lst}{}
\floatname{listing}{Listing}
\usepackage{algpseudocode}
\usepackage{multicol}
\usepackage{amsfonts}

\title{Intrinsic Action Tendency Consistency for 
Cooperative 

Multi-Agent Reinforcement Learning}
\author {
    Junkai Zhang\textsuperscript{\rm1,2},
    Yifan Zhang\textsuperscript{\rm 1,3,4}\thanks{Corresponding author},
    Xi Sheryl Zhang\textsuperscript{\rm 1,3,4},
    Yifan Zang\textsuperscript{\rm 1,2},
    Jian Cheng\textsuperscript{\rm 1,3,4}
}
\affiliations {
    \textsuperscript{\rm 1}
    Institute of Automation, Chinese Academy of Sciences\\
    \textsuperscript{\rm 2}School of Artificial Intelligence, University of Chinese Academy of Sciences \\ 
    \textsuperscript{\rm 3}University of Chinese Academy of Sciences, Nanjing \\     \textsuperscript{\rm 4}Nanjing Artificial Intelligence Research of AI \\

    \{zhangjunkai2021, xi.zhang, zangyifan2019\}@ia.ac.cn,
    \{yfzhang, jcheng\}@nlpr.ia.ac.cn
}

\begin{document}

\maketitle

\begin{abstract}
{Efficient collaboration in the centralized training with decentralized execution (CTDE) paradigm remains a challenge in cooperative multi-agent systems. We identify divergent action tendencies among agents as a significant obstacle to CTDE's training efficiency, requiring a large number of training samples to achieve a unified consensus on agents' policies. This divergence stems from the lack of adequate team consensus-related guidance signals during credit assignments in CTDE. To address this, we propose Intrinsic Action Tendency Consistency, a novel approach for cooperative multi-agent reinforcement learning. It integrates intrinsic rewards, obtained through an action model, into a reward-additive CTDE (RA-CTDE) framework. We formulate an action model that enables surrounding agents to predict the central agent's action tendency. Leveraging these predictions, we compute a cooperative intrinsic reward that encourages agents to match their actions with their neighbors' predictions. We establish the equivalence between RA-CTDE and CTDE through theoretical analyses, demonstrating that CTDE's training process can be achieved using agents' individual targets. Building on this insight, we introduce a novel method to combine intrinsic rewards and CTDE. Extensive experiments on challenging tasks in SMAC and GRF benchmarks showcase the improved performance of our method.}
\end{abstract}

\section{Introduction}
Cooperative multi-agent reinforcement learning (MARL) algorithms have shown the great capacity and potential to solve various real-world multi-agent tasks, such as automatic vehicles control \cite{sallab2017deep, zhou2020smarts}, traffic intelligence \cite{cao2012overview,  mushtaq2023multi, 
wiering2000multi}, resource management \cite{motlaghzadeh2023multi,sallab2017deep}, game AI \cite{berner2019dota, lin2023tizero} and robot swarm control \cite{dahiya2023survey, huttenrauch2017guided}. 
In a cooperative multi-agent system (MAS), every agent relies on their local observation to cooperate toward a team goal and the environment feedbacks a shared team reward.
There exist two major challenges in cooperative MAS: partial observability and scalability. Partial observability refers to the fact that agents can only access their local observations, resulting in unstable environments. 
Scalability refers to the challenge that the joint spaces of states and actions increase exponentially as the number of agents grows. 
To tackle these issues, \textit{Centralized Training with Decentralized Execution} (CTDE) paradigm is proposed {\cite{sunehag2017value}}, which allows agents to access the global state in the training stage and take actions individually. 
Given the CTDE paradigm, massive deep MARL methods have been proposed including VDN \cite{sunehag2017value}, QMIX \cite{rashid2020monotonic}, QTRAN \cite{son2019qtran}, QPLEX \cite{wang2020qplex} and so forth.
Their excellent performance can be attributed to the credit assignments, as rewards are critical as the most direct and fundamental instructional signals to drive behaviors \cite{silver2021reward, zheng2021episodic, mguni2021ligs}.

\begin{figure}[t]
  \centering
  \includegraphics[width=0.47\textwidth]{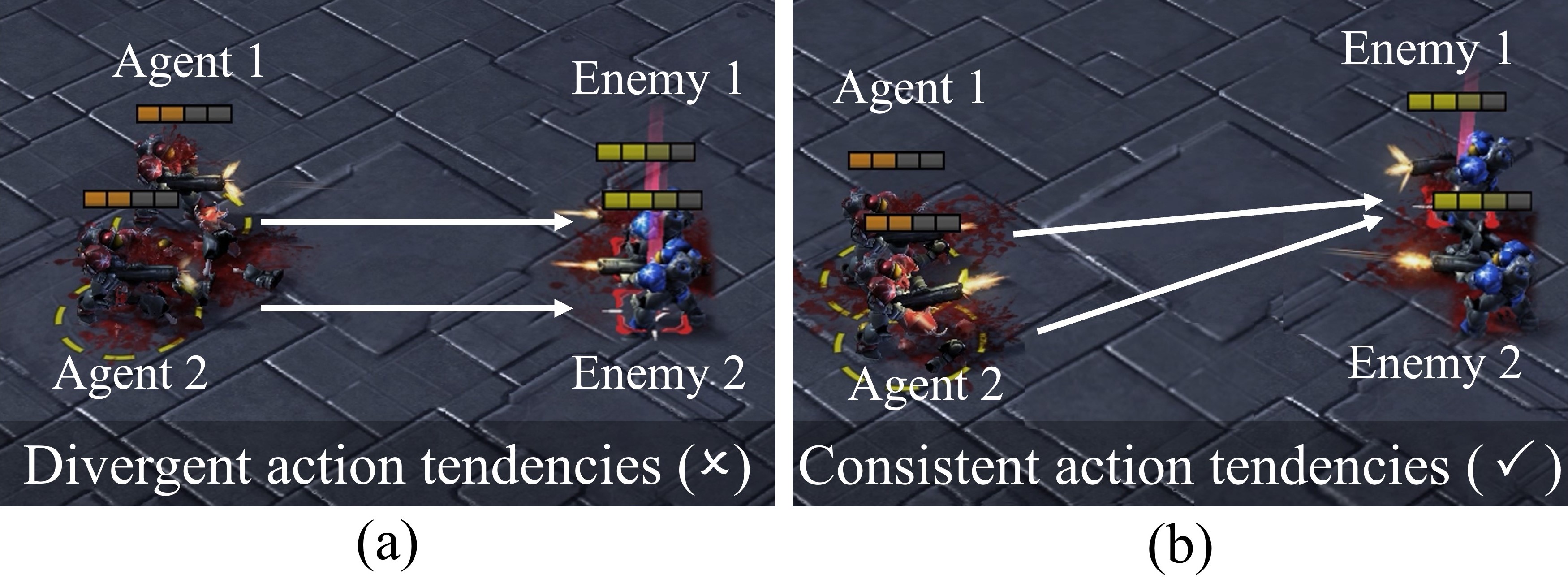}
  \caption{
  The illustration of the consistent action tendency.
  In (a) and (b), our agents' health value is lower than the enemies'. At this point,
  attacking either Enemy 1 or 2 simultaneously are the two best team policies. 
  In (a), Agent 1 and Agent 2 attack enemies separately without agreeing on a team policy.
  On the contrary, agents in (b) achieve a consistent goal policy and agree to attack a common enemy.
  To reflect the policy consistency among agents, we propose the concept of \textit{action tendency}. It reflects the policy distribution of agents toward different actions.
  We propose this \textit{action tendency} notion to distinguish it from \textit{policy}, which is usually the epsilon-greedy of $Q$ functions  only concerning the largest output in value-based approaches.
  }
  \label{motivation}
\end{figure}

{
However, it turns out that the sparse team rewards provided by many MAS environments cannot supply sufficient guidance for coordination behaviors, which results in inefficient training \cite{matignon2012independent}. 
We analyze the QMIX training process and realize that numerous unsuccessful episodes are caused by the inconsistency of team policy goals among agents like Figure \ref{motivation} (a). 
Among these episodes, each agent's action tendency is not unified to the same global policy.
We argue that the reward in MAS is the most essential instructional signal to drive behaviors and ascribe agents' action tendency inconsistency to CTDE's lack of sufficient team consensus guidance signals. 
}

An effective solution to this challenge is to add intrinsic rewards into the CTDE paradigm. 
There exist two major problems: \textit{how to design an intrinsic reward to guide agents' unified action tendency} and \textit{how to integrate the intrinsic rewards into the CTDE framework?}
In MARL, there are plenty of works designing intrinsic rewards including curiosity-based incentives \cite{ bohmer2019exploration,hernandez2019agent,iqbal2019coordinated, zhang2023self}, 
the mutual influence among agents \cite{chitnis2020intrinsic, jaques2019social,wang2019influence} and other specific designs \cite{strouse2018learning, ma2022elign, mguni2021ligs,du2019liir}.
However, {most of them are designed to enhance exploration } and employed in independent training ways, which suffer from unstable dynamics of environments.
To ease the latter problem,
EMC \cite{zheng2021episodic} proposed a curiosity-driven intrinsic reward and incorporated it into the CTDE training paradigm. Yet it averages the calculated intrinsic rewards and directly adds them to the global team reward, which results in losing the diversity of the intrinsic reward's adjustment toward credit assignments for each agent. 

In this work, we propose our novel {Intrinsic Action Tendency Consistency} for the cooperative multi-agent reinforcement learning method.
We hope to design intrinsic rewards on the basis of CTDE, so as to achieve consistent team policy goals among agents in the training process.
Specifically, we first propose an action model to predict the central agent's action tendency. 
We define our intrinsic reward as the surrounding agents' action tendency prediction error toward the central agents.
It encourages the central agent to take actions matching the prediction of their neighbors.
After that, we propose theoretical analyses on CTDE and convert it into an equivalent variant, RA-CTDE.
To appropriately utilize $N$ intrinsic rewards like IQL \cite{tan1993multi} training paradigm, we equivalently transform the original global {target} of CTDE into $N$ ones.
At last, we incorporate our action model based intrinsic reward into RA-CTDE and denote it by IAM.
We integrate our method into QMIX and VDN, and conduct extensive experiments in StarCraft II Micromanagement environment \cite{samvelyan2019starcraft} (SMAC) and  Google Football Research environment \cite{kurach2020google} (GRF).
Empirical results verify that our method achieves competitive performance and significantly outperforms other baselines.

\textbf{Key contributions} are summarized as follows: 1) We propose an action model based intrinsic reward measured by predicting the central agent's action tendency.
2) From a theoretical perspective, we address the issue of CTDE being unable to utilize the intrinsic rewards directly and consequently embed our intrinsic rewards into it.
3) By incorporating our method into  QMIX and VDN, we demonstrate IAM's competitive performance on challenging MARL tasks. 

\section{Background}

\subsection{Dec-POMDP}
A fully cooperative multi-agent task can be formulated as a \textit{Decentralized Partially Observable Markov Decision Process}  (Dec-POMDP) \cite{oliehoek2015concise},
which is an augmented POMDP formulated by a tuple $\mathfrak{M} \!= <\mathcal{N}, \mathcal{S}, (\mathcal{O}_i)_{i \in \mathcal{N}},(\mathcal{A}_i)_{i \in \mathcal{N}},  \mathbb{O},  \mathcal{P},  \mathcal{R}, \rho_{0}, \gamma \!>$, 
where every agent can only access the partial state of the environment and takes actions individually. 
Specifically,  
we denote $\mathcal{N} = \{1,..., N\}$ as the set of agents, where $N$ is the number of agents, $\mathcal{S}$ as the global finite state space, $\mathcal{O}_i$ as the partial observation of the state, obtained by the function $\mathbb{O}(s, i) |_{s\in \mathcal{S}}$,  and $\mathcal{A}_i$ as the action space respectively.
$\gamma \in [0, 1)$ is a discount factor and $\rho_0: \mathcal{S} \to R$  is the distribution of the inital state $s_0$.
The state transition probability function of the environment dynamics is $\mathcal{P}: \mathcal{S} \times \bm{\mathcal{A}} \times \mathcal{S} \to [0, 1]$ where $\bm{\mathcal{A}}:= \times_{i=1}^{{N}} \mathcal{A}_i$ is the joint action space selected by all agents. 
Due to the partial observable setting, every agent takes its observation-action history $\tau_i \in \{\mathcal{T}_i\}_{i=1}^{N} \equiv(\mathcal{O}_i\times \mathcal{A}_i)^* \times \mathcal{O}_i$ as the policy input to acquire more information. 
After agents taking their joint actions $\bm{a}:\{a_t^i\}_{i=1}^{{N}}$, the environment returns a team shared extrinsic reward $r^{ext} $ by function $\mathcal{R}(\mathcal{S}, \bm{\mathcal{A}})$: $\mathcal{S} \times \bm{\mathcal{\mathcal{A}}} \to \mathcal{R}$ . 
We define the stochastic policy of agent $i$ by $\pi_i(a_i|\tau_i): \mathcal{T}_i \times \mathcal{A}_i \to [0, 1]$, the multi-agent system algorithms are designed to find optimal policies $\pi^*=\{\pi_i^*\}_{i=1}^{{N}}$ to maximize the joint extrinsic value function $V^{\bm{\pi}}(s)=\mathbb{E}_{s_0, \bm{a_0},...}\left[\sum_{t=0}^{\infty}\gamma^t r^{ext}_t|{\bm{\pi}}\right]$, where $s_0 \sim \rho_0(s_0),  \bm{\pi}=\{\pi_i\}_{i=1}^{{N}}$.

\subsection{Centralized Training with Decentralized Execution}
The primary challenge for MAS tasks is that agents can only access partial observation and 
are incapable to acquire the global state, to which 
an effective solution is the CTDE training paradigm 
 {\cite{bernstein2002complexity}}. 
It allows all agents to access the global state in the centralized training stage and take actions individually in a decentralized manner. 
Formally, it formulates  $N$ individual $Q$-functions $\left\{Q_i(\tau_i,a_i; \theta_i)\right\}_{i\in \mathcal{N}}$ where $\theta_i$ is the network parameter for agent $i$.  
Meanwhile, it simultaneously preserves a joint action-value function $Q_{tot}(\bm{\tau},\bm{a})$ constructed by individual $ Q$ functions to help training.  
In detail, the objective of CTDE is to get an optimal joint action-value function  $Q^*_{tot}(s,\bm{a}) = r^{ext}(s, \bm{a}) + \gamma \mathbb{E}_{s'}[\max_{\bm{a'}}Q^*_{tot}(s',\bm{a}')] $. 
In the centralized training stage, $Q$-functions $\{Q_i\}_{i\in\mathcal{N}}$ are trained by minimizing the following {target} function:
\begin{align}
&\mathcal{L}^{G}\!(\bm{\theta}\!, \phi) \! = \!\! \mathbb{E}\left[r \!+\! \gamma \!\max_{\bm{a}'} {Q_{T}\!(\bm{\tau'\!,a'}})\! -\!Q_{tot}(\bm{\tau,a}; \bm{\theta}, \!\phi)\! \right ]^2 
\label{CTDE_LOSS} \\
&Q_{tot}(\bm{\tau,a}; \bm{\theta}\!, \phi)  =  \mathcal{F}\left(Q_1\!(\tau_1\!,a_1), \!..., \!Q_{{N}}(\tau_{{N}}\!,a_{{N}}), s; \phi \right)  
    \label{eq:Q_tot}
\end{align}
where $\bm{\tau}\!\!=\!\!\{\tau_i\}_{i=1}^{{N}}, \bm{a}\!\!=\!\!\{a\}_{i=1}^{{N}}, \bm{\theta}\!\! =\!\!\{\theta_i\}_{i=1}^{{N}}$, $\phi$ is the parameters of the mixing network  $\mathcal{F}$, and $\mathcal{D}$ is the replay buffer. $\{\bm{\tau, a, r, \tau'}\} \sim \mathcal{D}$. $Q_T$ denotes the expected return target function for the estimation of the global state-action pair. 
The gradients of $\bm{\theta}$ are calculated through function $\mathcal{F}$, which factorizes global $Q_{tot}$ function into decentralized ones $\{Q_i\}_{i=1}^{{N}}$, motivating enormous efforts to find factorization structures among them \cite{sunehag2017value, rashid2020monotonic,wang2020qplex}.

\begin{figure*}[ht!]
  \centering
  \includegraphics[width=0.95\textwidth]{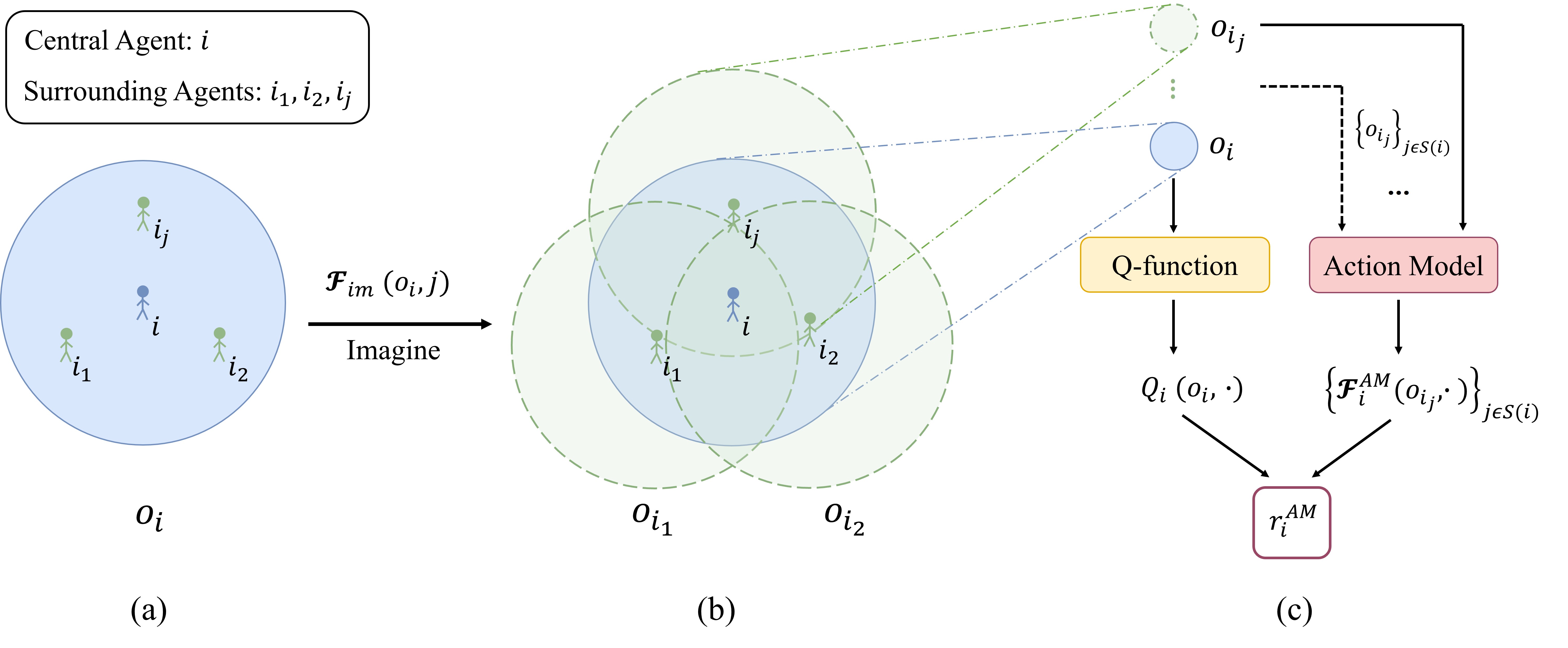}
  \caption{IAM-based reward. The blue and green zones represent the receptive field of the central agent and surrounding agents. The action model intrinsic reward is high when agent $i$ takes actions that match their surrounding agents' predictions.}
  \label{fig:am_rew}
\end{figure*}

\section{Related Works}
\label{Related Works}
 Many of the intrinsic reward functions used in MARL have been adapted from single agent curiosity-based incentives \cite{hernandez2019agent, iqbal2019coordinated, jaques2019social}, which aimed to encourage agents to explore their environment and seek out novel states. To better be applied in MARL, Some MARL-specific intrinsic reward functions have been proposed, including considering the mutual influence among agents \cite{chitnis2020intrinsic, jaques2019social,wang2019influence}, encouraging agents to reveal or hide their intentions \cite{strouse2018learning} and 
{predicting observation with alignment to their neighbors \cite{ma2022elign}.}
Besides,
Intrinsic rewards without task-oriented bias can increase the diversity of intrinsic reward space, which can be implemented by breaking the extrinsic rewards via credit assignment \cite{du2019liir} or using adaptive learners to obtain intrinsic rewards online \cite{mguni2021ligs}.
Apart from independent manners to dealing with rewards,
EMC \cite{zheng2021episodic} proposed a curiosity-driven intrinsic reward and introduce an integration way to accomplish the CTDE training paradigm.

\section{Method}

\label{others}
In this section, we present our Intrinsic Action Tendency Consistency for cooperative MARL denoted by IAM (Intrinsic Action Model).
Our purpose is to design an effective intrinsic reward to encourage consistent action tendencies and leverage it into CTDE in an {appropriate} manner. 
Specifically, we first introduce our action model based intrinsic reward, which encourages the central agent to {take actions consistent with its neighbors' prospects.}
Then we propose a reward-additive equivalent variant of the CTDE framework denoted by RA-CTDE  to incorporate our rewards reasonably.
At last, we analyze the essential difference between VDN \cite{sunehag2017value} and IQL \cite{tan1993multi} and then demonstrate the reasonability of our reward integration way.

\begin{figure*}[hbt!]
  \centering
  \includegraphics[width=0.9\textwidth]{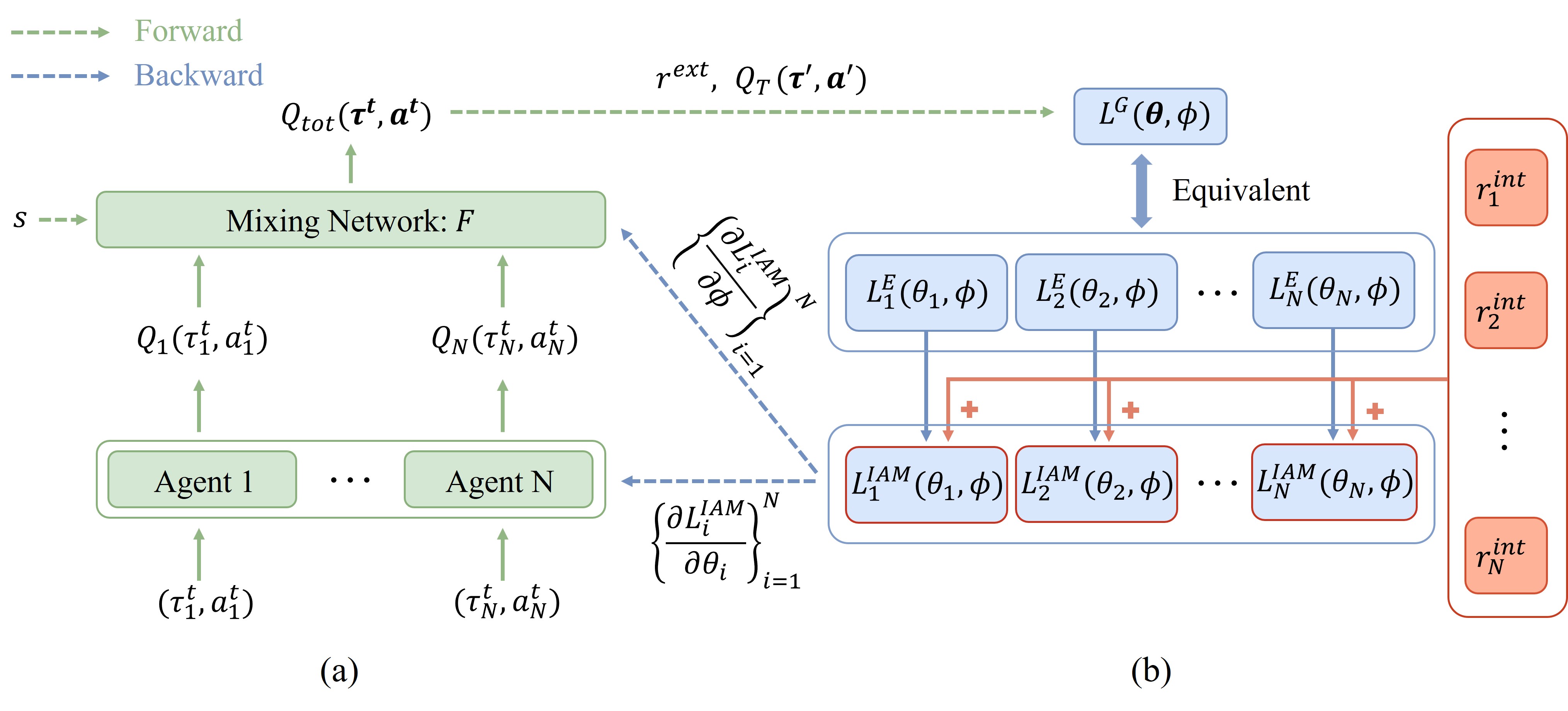}
  \caption{IAM training paradigm. The training paradigm consists of two stages: (a) Forward stage and (b) Backward stage.  In the forward stage, we use mixing network $\mathcal{F}, Q_{tot}, r^{ext}$ and $Q_T$ in Eq. \ref{CTDE_LOSS} to calculate global TD-error {target}  $\mathcal{L}^G$, which is the same as CTDE. In the backward stage, we first factorize $\mathcal{L}^G$ into ${N}$ targets $\{\mathcal{L}_i^E\}_{i=1}^{N}$ by Eq \ref{definition1} and \ref{definition2}, then add intrinsic rewards into them individually to obtain IAM {targets}: $\{\mathcal{L}_i^{IAM}\}_{i=1}^{N}$. The gradients of $\{\theta_i\}_{i=1}^{N}$ and $\phi$ are separately computed by backpropagating ${N}$ targets {$\{\mathcal{L}_i^{IAM}\}_{i=1}^{N}$}. } 
  \label{fig:IAM_method}
\end{figure*}

\subsection{Action Model Based Intrinsic  Reward}
For a better interpretation, we first give the following definitions: As shown in Figure \ref{fig:am_rew}, when considering a specific Agent $i$, we define it as the central agent. Due to the partial observability of the environment, agents that are observable in the surrounding area of Agent $i$ are defined as the surrounding agents and we denote the set as $S(i)$.
During the training process, we hope that every central agent $i$ will take into account its surrounding agents' expectations toward $i$'s {policy distribution}.
We denote its policy distribution by the action tendency, which represents the relative magnitude of an agent's inclination to take different actions.

\paragraph{Reward Calculation}
In a discrete action space, agent $i$'s action tendency can be reflected from two different perspectives.
From the viewpoint of the central agent $i$, its action tendency can be represented by its $Q_i$ function.
From the viewpoint of $i$'s surrounding agents, we define the action models $\{\mathcal{F}^{AM}_i\}_{i=1}^{N}$ to allow them to predict the central agent's action tendency.
$\mathcal{F}^{AM}_i$ is designed to utilize the same network structure as $Q_i$ function.
Their {representation distance} of action tendency reflects the central agent's consistency degree towards the surrounding agents' expectation.
Therefore we formulate this distance as an action model based intrinsic reward, i.e. $\{r_i^{AM}\}_{i=1}^{N}$.
\begin{align}
&o_{i_j} = \mathcal{F}_{{im}}(o_i, j) \label{eq_fn} \\
r^{AM}_i\! =\! \!\frac{-1}{|{S}(i)|}\!\!\sum_{j\in {S}(i)}&\!\!\!Dis\!\left(\mathcal{F}^{AM}_i\!(\!o_{i_j}, \!\cdot\ ; \omega_i\!)\! - \!Q_i(o_i, \cdot )\right) \label{eq_ram}
\end{align}
The reward calculation process is illustrated in Figure \ref{fig:am_rew} and Eq \ref{eq_fn}, \ref{eq_ram}.
During the training phase, every agent first calculates its imagined surrounding agents' observations, then utilizes its $Q$ function and action model to measure the action tendency distance, and finally obtains its action model based intrinsic reward $r_i^{AM}$.
The {imagine function} $\mathcal{F}_{im}(o_i, j)$ in Eq \ref{eq_fn} is defined to represent the surrounding agents' simulated observation \textit{imagined} by the central agent $i$.
The \textit{imagining} process is realized by switching the viewpoints from the central agent into the surrounding agents, i.e., separately setting the positional coordinates of every surrounding agent as the origin to calculate the coordinates of the other agents attached with additional information, which does not require any learning parameters (more details in the Appendix ).
In experiments, we use the $L_2$ distance as the $Dis$ function.
Under this reward setting, agents are encouraged to take actions consistent with their surrounding agents' prospects.
\begin{equation}
\begin{aligned}
Q_{i}^{(t+1)}\!\left(\tau_{i}, \!a_{i}\right)\!=\!&\underbrace{\underset{\left(\tau_{-i}^{\prime}, a_{-i}^{\prime}\right) \sim p_{D}\left(\cdot \mid \tau_{i}\right)}{\mathbb{E}}\!\left[y^{(t)}\!\left(\tau_{i} \oplus \tau_{-i}^{\prime}, a_{i}\! \oplus \!a_{-i}^{\prime}\right)\!\right]}_{\text {evaluation of the individual action } a_{i}} \\
&-\underbrace{\frac{n-1}{n} \underset{\boldsymbol{\tau}^{\prime}, \boldsymbol{a}^{\prime} \sim p_{D}\left(\cdot \mid \Lambda^{-1}\left(\tau_{i}\right)\right)}{\mathbb{E}}\left[y^{(t)}\left(\boldsymbol{\tau}^{\prime}, \boldsymbol{a}^{\prime}\right)\right]}_{\text {counterfactual baseline }}\\ 
&+\underbrace{w_{i}\left(\tau_{i}\right)}_{\text {residue term}} 
\end{aligned}\label{long_formulation}
\end{equation}
\paragraph{Action Model Training}
To obtain $\mathcal{F}^{AM}$, we use $Q_i$ function values as supervised targets.
This choice is reasonable based on the following insight: 
\textit{The individual $Q_i$ value incorporates interaction information of other agents to agent $i$, not just only the agent $i$'s own action tendencies with linear value factorization.} 
In VDN training paradigm, the individual $Q_i$ function can be factorized into Eq \ref{long_formulation}'s form \cite{wang2020towards}, where $p_{D}(\cdot | \tau_i)$ denotes the conditional empirical probability of $\tau_i$ in the given dataset $D$ ,
the notation $\tau_i \bigoplus \tau_{-i}^{\prime}$ denotes $<\tau_{1}^{\prime}, ... , \tau_{i-1}^{\prime}, \tau_{i}, \tau_{i+1}^{\prime},...,\tau_{n}^{\prime}>$,
and $\tau_{-i}^{\prime}$ denotes the elements of all agents except for agent $i$.

In Eq \ref{long_formulation},
it is easy to see that the $Q_i$ function essentially consists of three items, and the first two
include the expectation of one-step TD target value over others.
It indicates that the $Q_i$ function value obtained in VDN includes the interactive historical expectation toward other agents.
Although this analysis only applies to VDNs, we broaden the supervised target $Q$ functions to QMIX and also achieve effective performance improvement.
The pseudo-code of our algorithm is interpreted in the Appendix.

\subsection{Reward-Additive CTDE (RA-CTDE)}
The contradiction for CTDE to {utilize} $N$ intrinsic rewards is that it has only one global {target} $\mathcal{L}^G$  during training.
However, 
IQL \cite{tan1993multi} can directly use ${N}$ different intrinsic rewards naturally because it obtains $N$  TD-losses individually. 
Based on that, we first factorize the global target $\mathcal{L}^G$ in Eq \ref{CTDE_LOSS} into ${N}$ individual ones and define it as Reward-Additive CTDE (RA-CTDE).
Then we demonstrate its equivalence with the original target $\mathcal{L}^G$.
At last, we discuss how to add intrinsic rewards to the RA-CTDE.
\begin{definition}
    (Reward-Additive CTDE).  Let $\bm{\theta} =\{\theta_i\}_{i=1}^{{N}}$ be the parameters of $Q$ functions, $\mathcal{F}$ be the mixing network in CTDE, $\mathcal{N}\!=\!\{1,...,N\}$ be the agents set, $\mathcal{Q}^{N}\!\!=\!\{Q_1(\tau_1,a_1; \theta_1), Q_2(\tau_2,a_2; \theta_2),..., Q_{{N}}(\tau_{{N}},a_{{N}}; \theta_{{N}})\}, \{\bm{\tau},\bm{a},\\{r^{ext}}, \bm{\tau^{\prime}}\} \!\sim\! \mathcal{D}$, assume $\forall i,j \in \mathcal{N}, \theta_i \neq \theta_j$, then Reward-Additive CTDE means computing $\{\mathcal{L}_i^{E}(\theta_i, \phi)\}_{i=1}^{N}$ in Eq \ref{definition1} and Eq \ref{definition2}  and then updating their parameters respectively. The term $\mathcal{P}$ is not involved in the gradient calculation as a scalar. Formally{:}
\begin{align}
\mathcal{L}_i^{E}(\theta_i, \phi)
  =\mathbb{E}_{\bm{\tau},\bm{a},{r^{ext}}, \bm{\tau^{\prime}}\in \mathcal{D}}
    \left[ \mathcal{P}
\cdot 
\mathcal{F}(\mathcal{Q}^{{N}},s ; \phi)
\right] 
\label{definition1} \\
\mathcal{P}=-2\left(r^{ext}+\gamma \max_{\bm{a}'} {Q_{T}(\bm{\tau',a'}}) -\mathcal{F}(\mathcal{Q}^{{N}},s ; \phi)\right)
 \label{definition2}
\end{align} 
\end{definition}

We propose our reward-additive variant of the CTDE framework RA-CTDE, where the $Q_T(\bm{\tau^{\prime}},\bm{a^{\prime}})$ in Eq \ref{definition2} is updated by the target function of the mixing network $\mathcal{F}$.\footnote{Please note that although $\mathcal{L}_i^{E}$ in \ref{definition1} is the same across different agents, their corresponding computed gradients are different, which is detailed in the Appendix.}
We consider that RA-CTDE is equivalent to the original CTDE paradigm based on the following theorem.
\begin{theorem}
Let $\{\theta_i\}_{i=1}^{{N}}$ be the parameters of $Q$ functions, $\phi$ be the parameters of the mixing network $\mathcal{F}$ in CTDE, $\mathcal{L}^G$ be the global target in Eq \ref{CTDE_LOSS}, $\mathcal{N}=\{1,...,N\}$ be the agents set, $\mathcal{Q}^{N}\kern-9pt=\kern-9pt\{Q_1(\tau_1,a_1; \theta_1), Q_2(\tau_2,a_2; \theta_2),..., Q_{{N}}(\tau_{{N}},a_{{N}}; \theta_{{N}})\}$,  $\{\bm{\tau},\bm{a},{r^{ext}}, \bm{\tau^{\prime}}\} \kern-2pt\sim\kern-2pt \mathcal{D}$, assume $\forall i,j \in \mathcal{N}, \theta_i \neq \theta_j$, then $\forall i \in \mathcal{N}$ , the following equations hold{:}
\begin{align}
\frac{\partial \mathcal{L}^{G}(\bm{\theta}, \phi)}{\partial \theta_i} &= \frac{\partial \mathcal{L}_i^{E}({\theta_i, \phi})}{\partial \theta_i}
\label{theorema} \\
\frac{\partial \mathcal{L}^{G}(\bm{\theta}, \phi)}{\partial \phi} &= \frac{1}{N}  \sum_{i=1}^{N} \frac{\partial \mathcal{L}_i^{E}({\theta_i, \phi})}{\partial \phi}
 \label{theoremb}
\end{align} 
\label{theorem}
\end{theorem}
The Theorem \ref{theorem} is proved in the Appendix. According to it, we draw the conclusion that the CTDE's essence in updating gradients of $\{\theta_i\}_{i=1}^{{N}}$ and $\phi$ is to calculate the global target $\mathcal{L}^G$ and then respectively perform $N$  gradient backpropagation steps for each agent.
Therefore we can equivalently factorize the global target $\mathcal{L}^G$ into $N$ individual ones denoted by $\mathcal{L}^E_i$  in RA-CTDE. 
The factorized target $\mathcal{L}^E_i$ provides an interface for adding rewards and we exhibit the reward-adding way based on the following corollary.
\begin{corollary}
Let $\{\theta_i\}_{i=1}^{{N}}$ be the parameters of $Q$ functions,  $\mathcal{N}\!=\!\{1,...,N\}$ be the agents set,  $\mathcal{L}_i^{VDN}$ be the $\mathcal{L}_i^{E}$ 's special case of VDN, $\{\bm{\tau},\bm{a},{r^{ext}}, \bm{\tau^{\prime}}\} \!\!\sim\!\! \mathcal{D}$, assume $\forall i,j \in \mathcal{N}, \theta_i \neq \theta_j$,  
 then $\forall i \in \mathcal{N} ${:} 
\begin{align}
 &\mathcal{L}_{i}^{VDN}(\bm{\theta})=\mathbb{E}_{\bm{\tau},\bm{a},{r^{ext}}, \bm{\tau^{\prime}}\in \mathcal{D}}\left[ \mathcal{P}_i \cdot Q_i({\tau_i,a_i}; \theta_i)\right]
\label{eq_corollary1} \\
&\mathcal{P}_i\!\!=\!\!-2\!\left(\!r^{ext}\!\!+\!  \mathcal{R}_i^{VDN} \!\!\!+\! \gamma\! \max_{a'} {Q^{-}_i\!({\tau_i',a_i'}}) \!\!-\! Q_i({\tau_i,a_i})\! \right)
\label{eq_corollary2}\\
&\mathcal{R}_i^{VDN} \!=\! \gamma \!\max_{a'}\!\!\sum_{j=1,j \neq i}^{{N}} \!\!{Q^{-}_j({\tau_j', a_j'}}) \!-\!\!\!\! \sum_{j=1, j\neq i}^{{N}}\!Q_j({\tau_j, a_j})
\label{eq_corollary3}
\end{align}
\label{corollary}
\end{corollary}
\begin{figure*}[htbp!]
  \centering
\includegraphics[width=\textwidth]{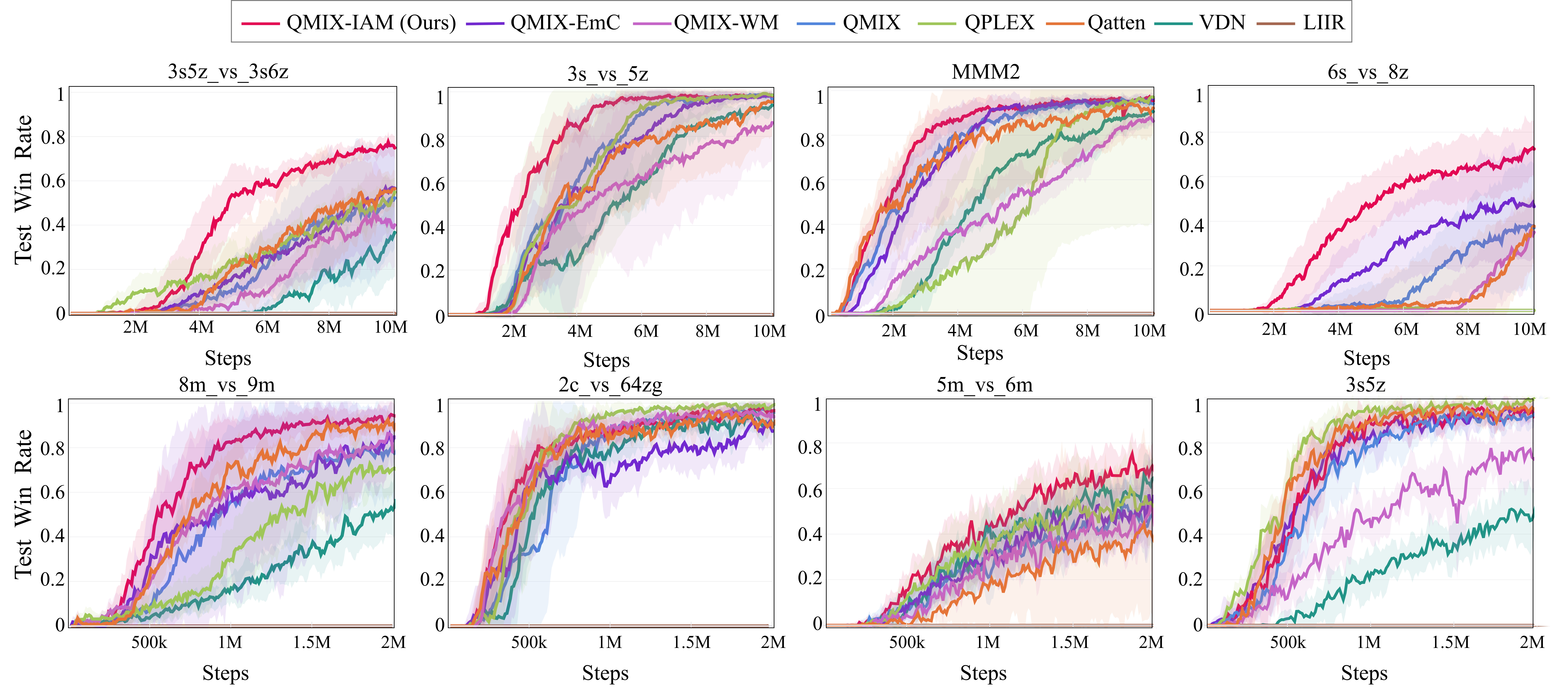}
  \caption{Performance comparisons for various maps in SMAC.}
  \label{fig:SMAC_exp}
\end{figure*}
We consider the special case VDN \cite{sunehag2017value} and transform it into the RA-CTDE form, where the mixing network $\mathcal{F}$ is the calculation of summing over all $Q_i$ functions. 
On the basis of the Eq \ref{eq_corollary1}, \ref{eq_corollary2}, and \ref{eq_corollary3} in Corollary \ref{corollary}, we realize that VDN can be factorized into $N$ targets like IQL\cite{tan1993multi}. 
But the essential difference between VDN and IQL \cite{tan1993multi} is that the former adds  certain intrinsic rewards $\{\mathcal{R}_i^{VDN}\}_{i=1}^{N}$ into $\{\mathcal{P}_i \}_{i=1}^{N}$.
In other words, when the $\mathcal{R}_i^{VDN}$ are not incorporated in Eq \ref{eq_corollary2}, VDN fundamentally boils down to IQL. 
The reward $\mathcal{R}_i^{VDN}$ is incorporated into the TD-error term $\mathcal{P}_i$.
We adopt the same reward-adding form as VDN and extend it to the RA-CTDE framework.
Specifically, we choose to add the calculated intrinsic rewards with parameter $\beta$ into ${N}$  losses $\{\mathcal{L}_i^{E}\}_{i=1}^{{N}}$ and get our IAM {targets} in Eq \ref{eq:IAM_loss1}, \ref{eq:IAM_loss2}.
The gradient of $\theta_i$ and $\phi$ can be obtained by computing $\frac{\partial \mathcal{L}_i^{IAM}({\theta_i, \phi})}{\partial \theta_i}$ and $\frac{1}{N} \cdot \sum_{i=1}^{N} \frac{\partial \mathcal{L}_i^{IAM}({\theta_i, \phi})}{\partial \phi}$ respectively.
Figure \ref{fig:IAM_method} shows our whole training paradigm.
\begin{align}
&\mathcal{L}_i^{IAM}(\theta_i, \phi)
=\mathbb{E}_{\bm{\tau},\bm{a},{r^{ext}}, \bm{\tau^{\prime}}\in \mathcal{D}}
    \left[\mathcal{P}_i \cdot \mathcal{F}(\mathcal{Q}^{{N}},s;\phi)\right]
\label{eq:IAM_loss1} \\
&\mathcal{P}_i \! = \!\!-\!2\!\left(\!r^{ext}\!\!+\!\beta {r^{int}_i}\!\!+\!\gamma \!\max_{\bm{a}'}\! {Q_{T}(\bm{\tau'\!,\!a'}})\!\! -\!\! \mathcal{F}(\mathcal{Q}^{{N}}\!\!,s;\! \phi)\!\right)
 \label{eq:IAM_loss2}
\end{align}
{Though the training paradigm of IAM also uses N targets motivated by the reward-adding way like IQL,
the target $\mathcal{L}_i^{IAM}$ still contains other agents' information 
and the essence of IAM is an improved CTDE instead of an independent training method.}

\section{Experiments}
To demonstrate the high efficiency of our algorithm, we exploit different environments to conduct a large number of experiments, including StarCraft II Micromanagement (SMAC) \cite{samvelyan2019starcraft}, Google Research Football (GRF) \cite{kurach2020google} and  Multi-Agent Particle Environment (MPE) \cite{mordatch2018emergence}. 
We conduct 5 random seeds for each algorithm and report the 1st, median, and 3rd quartile results.
Due to space limitations on pages, we leave the MPE experiments in the Appendix.

\subsection{Experiments Setup}
\paragraph{StarCraft II Micromanagement}
The StarCraft Multi-Agent Challenge \cite{samvelyan2019starcraft} is a popular benchmark in cooperative multi-agent environments, where  agents must form groups and work together to 
attack built-in AI enemies.
The controlled units only access local observations within a limited field of view and take discrete actions.
At each time step, every agent takes an action, and then the environment feedbacks a global team reward, which is computed by the accumulative damage point to the enemies. 
To evaluate the efficacy of different algorithms, we employ the training paradigm as previous notable works \cite{du2019liir, zhou2020learning} which utilizes 32 parallel runners to generate trajectories and store them into batches.

\paragraph{Google Research Football}
\label{head_GRF_exp}
The Academy scenarios of the Google Research Football environment \cite{kurach2020google} are inherently cooperative tasks that simulate partial football match scenes.
We use the \textit{Floats} (115-dimensional vector) observation setting including \textit{players' coordinates, ball possession, ball direction, active player, and game mode.} 
The GRF is a highly sparse reward benchmark because it only feedbacks a global team reward $r$ in the end, i.e., +1 bonus when scoring a goal and -1 bonus when conceding one.


\subsection{Performance Comparisons}
\label{head_samc_exp1}
To demonstrate the effectiveness of IAM, we combine it with two representative  CTDE algorithms: QMIX and VDN, which represent two ways of value factorization, i.e., linear and non-linear.
We denote them as  QMIX-IAM and VDN-IAM respectively.
To compare the performance of different rewards, we choose  different types of reward-shaping methods as baselines: 
(1) Curiosity-based intrinsic rewards in EMC \cite{zheng2021episodic}.
It's a representative CTDE's reward-shaping method based on curiosity. 
To fairly compare the impact of rewards, we remove its episodic memory and incorporate it with VDN and QMIX  denoted by VDN-EmC and QMIX-EmC respectively. 
(2) Add world model based reward \cite{ma2022elign} into RA-CTDE, denoted by VDN-WM and QMIX-WM. 
(3) LIIR \cite{du2019liir} that utilizes learned intrinsic rewards.
The remaining baselines are CTDE algorithms:
(4) QPLEX.
(5) Qatten.
(6) QMIX.
(7) VDN.
For ease of comparison, we separate the performance comparison of QMIX-IAM and VDN-IAM, details on the latter are demonstrated in the Appendix.

\paragraph{QMIX-IAM outperforms baselines.}
As shown in Figure \ref{fig:SMAC_exp}, the performance of QMIX has been significantly improved after using the action model based reward, and QMIX-IAM outperforms other baselines in most scenarios, especially on several very hard maps requiring strong team cooperation. It indicates that the action model based reward can encourage consistent policy behaviors among agents and improve the performance of the CTDE algorithm.
When using the exploration-based reward alone, QMIX-EmC only achieves performance improvements over QMIX on \textit{6h\_vs\_8z} and \textit{3s\_vs\_5z}, which indicates that the exploration-based reward lacks generalization for cooperative tasks.
Based on the world model intrinsic reward, the performance of QMIX only has performance improvements on \textit{3s5z}. This indicates that the world model based reward cannot generalize well in complex scenarios for high-dimensional observations and  lacks in reflecting agents'  {action tendencies}.
Besides, QMIX-IAM also performs better than prominent CTDE methods, i.e. QPLEX and Qatten.

\begin{figure*}[htb]
  \centering
  \includegraphics[width=\textwidth]{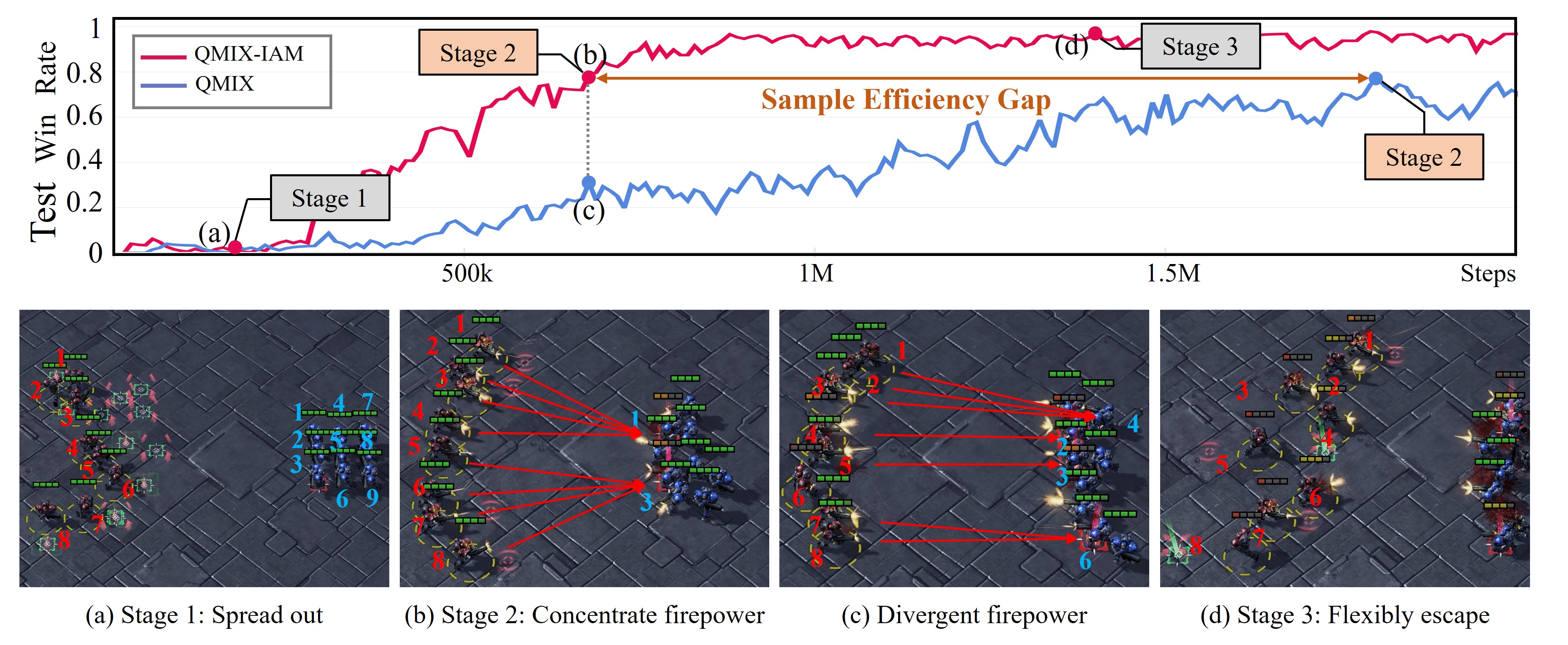}
  \caption{
  A visualization example of IAM on 8m\_vs\_9m. 
  In this task, agents need to obtain a three-stage team policy to win.
  In stage 1, agents need to be spread out to maximize the distraction of enemy attacks.
  In stage 2, agents need to maximize the concentration of firepower on the same enemy and reduce the enemy's numbers.
  In stage 3, agents need to escape quickly when they are low on blood to avoid being attacked and increase survival time. 
  Among them, stage 2 is the hardest to learn because the agents need to cooperate to achieve the same policy target, i.e. action tendency consistency.
  (a), (b), and (d) represent  three team policy stages of QMIX-IAM. 
  (d) exhibits the distributed fire against the enemy of QMIX.}
  
  \label{fig:visualization}
\end{figure*}

\begin{figure*}[htb!]
  \centering
  \includegraphics[width=\textwidth]{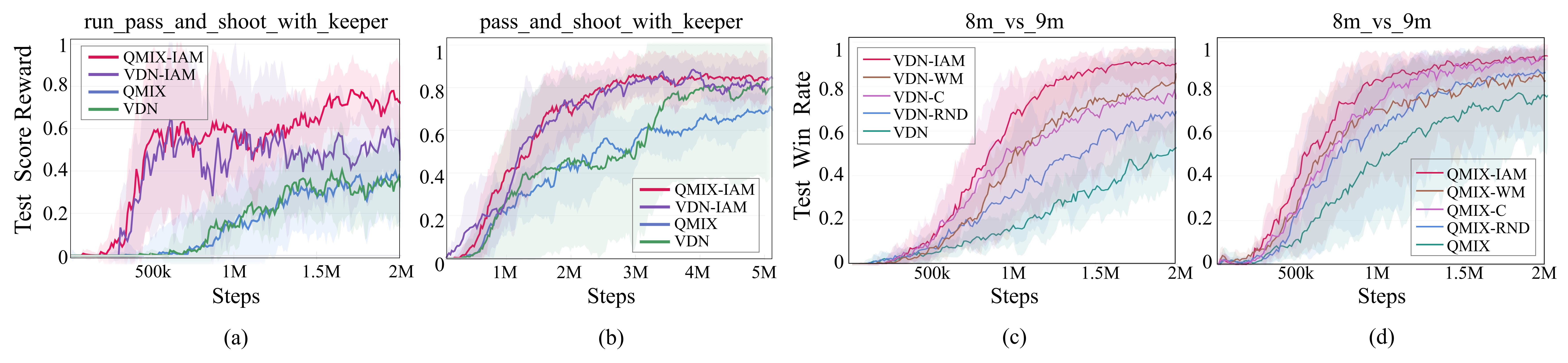}
  \caption{Ablation experiments. (a) and (b) show the performance comparison in two scenes of GRF. (c) and (d) 
  exhibit the performance comparison in RA-CTDE combined with different rewards.
  }
  
  \label{fig:ablations}
\end{figure*}

\subsection{Strengths of IAM}

\paragraph{{An explicable example of IAM's impact.}}
We visualize an illuminating map 8m\_vs\_9m in Figure \ref{fig:visualization}, demonstrating how the 
IAM improves QMIX's performance by action tendency consistency.
Among 3 training policy stages, QMIX takes the most samples to achieve  Stage 2. The essential reason is that agents cannot reach team unanimity when attacking, causing their dispersion of firepower. 
Under the guidance of our action model intrinsic reward, agents will take the initiative to cultivate a tacit understanding of each other's action tendencies. 
Then agents can  quickly achieve the team's consistent goal with only a few training samples, thus greatly improving sample efficiency than QMIX.

\paragraph{IAM can also obtain improved performance in highly sparse reward environments.}
To evaluate IAM's performance in deeply sparse reward environments, we choose two challenging tasks from GRF including 
\textit{Academy\_run\_pass\_and\_shoot\_with\_keeper} and   \textit{Academy\_pass\_and\_shoot\_with\_keeper}.
We choose QMIX and VDN as baselines.
As shown in Figure \ref{fig:ablations} (a) and (b), 
our method can significantly enhance the performance of QMIX and VDN, which indicates that IAM generalizes well in sparse-reward environmental scenarios.

\paragraph{Ablation: Our proposed intrinsic reward outperforms others when using RA-CTDE.}
In order to compare the performance of different rewards added in RA-CTDE, we compare IAM with additional  baselines:
(1) Add curiosity based  reward in EMC, denoted by VDN-C and QMIX-C. 
(2) Add random network distillation(RND) reward into RA-CTDE, denotedy VDN-RND and QMIX-RND.
We conduct these algorithms in \textit{8m\_vs\_9m} and demonstrate results in Figure \ref{fig:ablations}.  
Both VDN-IAM and QMIX-IAM outperform others which implies that predictive information about action is beneficial for cooperation. Besides, after using RA-CTDE, all these intrinsic rewards have achieved performance improvements, indicating that the way RA-CTDE uses intrinsic rewards is reasonable and provides a new direction for CTDE to utilize intrinsic rewards.
Besides, Exploration-based rewards don't perform as well as the action model based rewards, which indicates that 
the RA-CTDE framework can use different intrinsic rewards but the cooperative intrinsic rewards perform better than the exploration one in cooperative multi-agent systems. 

Besides the aforementioned experiments, we also conduct ablation experiments to demonstrate {the outperformance of RA-CTDE's reward-adding manner} and {RA-CTDE's equivalence to CTDE}, which are detailed in the Appendix.


\section{Conclusions and Limitations}
We find that the CTDE algorithm suffers from low sample efficiency and attribute it to the team consensus inconsistency among agents.
To tackle this problem, we design a novel intrinsic action model based reward and 
transform the CTDE into an equivalent variant, RA-CTDE.
Then we use a novel integration of intrinsic rewards with RA-CTDE. 
Since our action model intrinsic rewards can boost consistent team policy and our proposed RA-CTDE can flexibly use calculated intrinsic rewards, our method shows significant outperformance on challenging tasks in the SMAC and GRF benchmarks. 
The limitations of our work are that we did not consider environments with continuous state-action space and did not make specific designs for heterogeneous agents. 
For future work, we will conduct additional research in the aforementioned directions.

\section*{Acknowledgements}
This work was supported in part by the National Key R\&D Program of China (2022ZD0116402), NSFC 62273347, Jiangsu Key Research and Development Plan (BE2023016).


\bibliography{reference}

\clearpage

\appendix

\section{Algorithm Details}
\subsection{Pseudo Code}
IAM's pseudo code is shown in Algorithm \ref{algorithm1}.

\begin{algorithm}[h]
\caption{Intrinsic Action Tendency Consistency for cooperative multi-agent reinforcement learning}
\begin{algorithmic}[1]
\State Initialize the agents' $Q$ functions $\{Q_i\}_{i=1}^{{N}} $ with parameters $\bm{\theta} \!\!=\!\! \{\theta_i\}_{i=1}^{{N}}$,
\! $\mathcal{A}$ction models $\{\mathcal{F}^{{AM}}_i\}_{i=1}^{{N}} $ with parameters $\{\omega_i\}_{i=1}^{{N}}$, 
CTDE mixing network $\mathcal{F}(Q_1, Q_2,...,Q_N,s)$ with parameters $\phi$, replay buffer $\mathcal{D}$, number of mini-batch $N_B$, 
 rollout numbers $N_r$, rollout length $T$, discount factor\ $\gamma$, $Q$ function learning rate $\alpha$, mixing network $\mathcal{F}$ learning rate $\nu$, 
$\mathcal{A}$ction model learning rate $\eta$, 
intrinsic reward coefficient $\beta$, environment $E$, 
CTDE target network $Q_T(s,a)$, 
transitions $\left(s^t, \{o_i^t\}_{i=1}^{N}, \{a_i^t\}_{i=1}^{N}, r_{ext}^t, s^{t+1} \right)$.
\For{$n=1$ to $N_r$}
    \State \textbf{// Collect Rollouts}
    \For{$t=1$ to $T$}
        \State Sample global states $s^t$ from environment $E$
        \State Get observations $\{o_i^t\}_{i=1}^{N}$ by function $\mathbb{O}(s^t,i)$
        \State Sample actions $\{a_i^t\}_{i=1}^{N}$ from  $\{(Q_i^t(o_i^t, \cdot)\}_{i=1}^{N}$
        \State Take actions then sample $r_{ext}^t$ and $s_{t+1}$ from $E$
        \State Append transitions to replay buffer $\mathcal{D}$
        \State \textbf{// Train Agents and Action Models}
        \State Sample $N_B$ transitions
        \For{$i=1$ to $N$}
            \State \textbf{// Compute Action Model Rewards}
            \State Compute $\{o_{i_j}\}_{j=1}^{|\mathcal{S}(o_i)|}$ using $\mathcal{F}_{im}(o_i, j)$ 
            \State Compute  $r^{AM}_i$  using Eq \ref{app_eq_ram}
            \State Compute  $\mathcal{L}^{AM}_i$ using Eq \ref{action model loss}
            \State  $\phi_i \leftarrow \phi_i - \eta\nabla \mathcal{L}_i^{AM}$
            \State \textbf{// Train Agent Policy}
            \State Compute IAM Loss $\mathcal{L}_i^{IAM}$ using Eq \ref{eq:app_IAM_loss1}, \ref{eq:app_IAM_loss2}
            \State $\theta_i \leftarrow \theta_i - \alpha\nabla \mathcal{L}_i^{IAM}$
            \State $\omega_i \leftarrow \omega_i - \nu\nabla \mathcal{L}_i^{IAM}$
        \EndFor   
        \State Update target network $Q_T(s, a)$
    \EndFor
\EndFor
\end{algorithmic}

\label{algorithm1}
\end{algorithm}

The $o_{i_j}$ is obtained by the \textit{imagine function}:
\begin{equation}
o_{i_j} = \mathcal{F}_{{im}}(o_i, j) \label{app_eq_fn}
\end{equation}

The intrinsic reward is obtained by the action model prediction error:
\begin{equation}
r^{AM}_i = -\frac{1}{|\mathcal{S}(i)|}\sum_{j\in \mathcal{S}(i)}^{} Dis(\mathcal{F}^{AM}_i(o_{i_j}, \cdot\ ; \omega_i) - Q_i(o_i, \cdot\ ;\theta_i)) \label{app_eq_ram}
\end{equation}

The action model is trained to minimize the loss function:
\begin{equation}
\mathcal{L}_{i}^{AM} = \sum_{i=1}^{N} ||\mathcal{F}_i^{AM}(o_{i}, \cdot; \omega_i)-Q_i(o_i, \cdot; \theta_i)||_2^2 \label{action model loss}
\end{equation}

$Q$ functions and mixing network $F$ are trained to minimize the loss function:

\begin{align}
  &\mathcal{L}_i^{IAM}(\theta_i, \phi)
    =\mathbb{E}_{\bm{\tau},\bm{a},{r^{ext}}, \bm{\tau^{\prime}}\in \mathcal{D}}
      \left[ \mathcal{P}_i
  \cdot 
  \mathcal{F}(\mathcal{Q}^{{N}},s ; \phi)
  \right] 
  \label{eq:app_IAM_loss1} \\
  &\mathcal{P}_i\!=\!-2\left(r^{ext}+\!\beta r_i^{int}\!+\!\gamma\! \max_{\bm{a}'} {Q_{T}(\bm{\tau',a'}}) \!-\!\mathcal{F}(\mathcal{Q}^{{N}},s ; \phi)\right)
  \label{eq:app_IAM_loss2}
\end{align} 

In our IAM method, we set $Dis$ function as $L2$ distance and set  $r^{int}_i$ as $r^{AM}_i$. 
Besides, as Eq \ref{eq:app_IAM_loss1} and \ref{eq:app_IAM_loss2} indicate, the intrinsic reward can be independently added into each agent-associated loss function readily since the Theorem \ref{theorem} facilitates the factorization of the global loss function.

\subsection{Imagined Observation} \label{imagined_observation}
    \begin{figure*}[htb]
     \centering
     \includegraphics[width=0.7\textwidth]{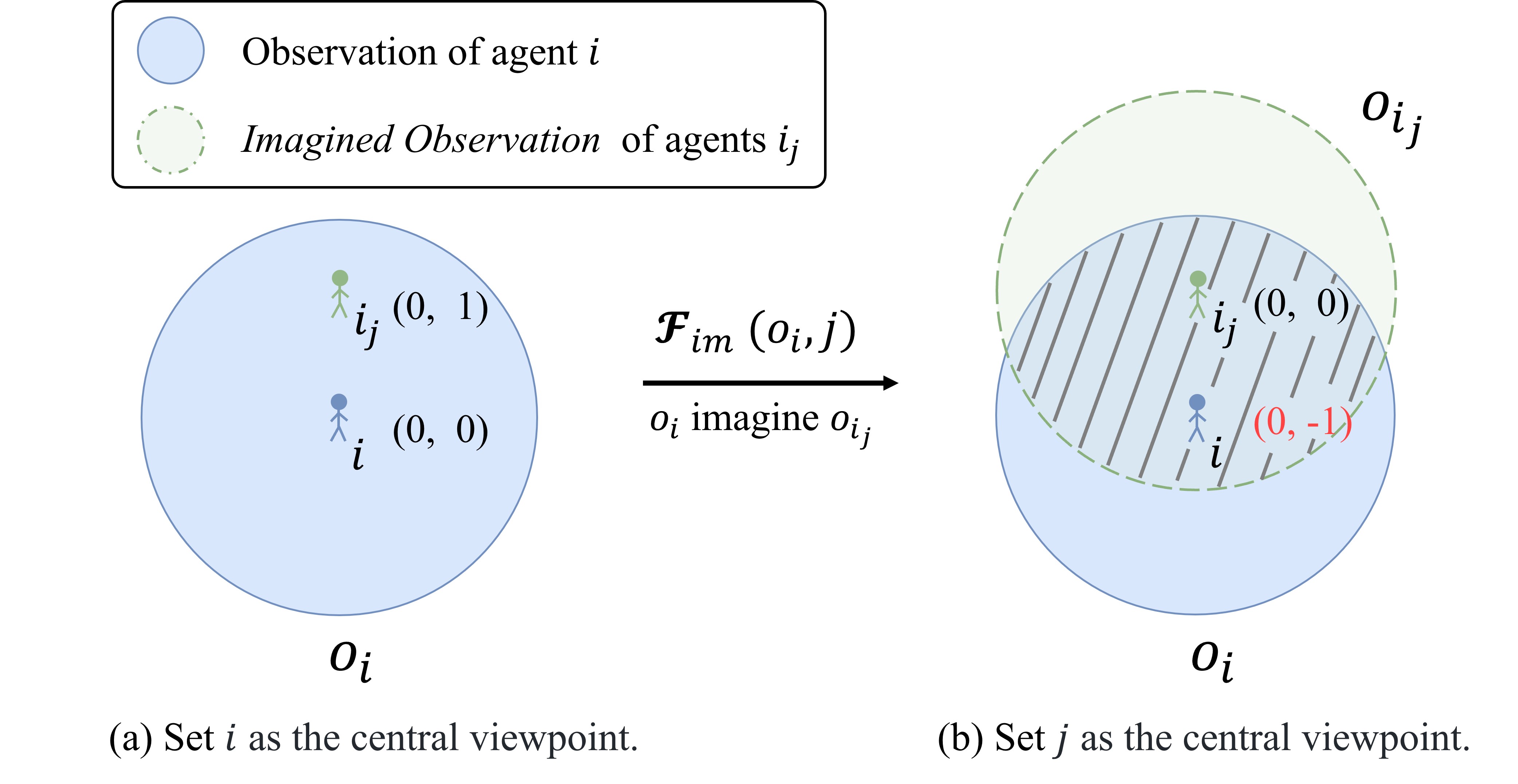}
     \caption{An example of the \textit{Imagined Observation} calculation. We only use the agent $i$'s observation to calculate the imagined observation of agent $j$, by switching the central viewpoint. In the $i$'s viewpoint, the $i$'s coordinate position is (0, 0), and the $i_j$'s coordinate position is (0, 1). When we set $j$ as the central viewpoint, the $j$'s coordinate position is (0, 0), then we can calculate the $i$'s coordinate position as \textbf{(0, -1)},  which represents a part of $o_{i_j}$ which stands for agent $j$'s observation from agent $i$'s imagination, shown in the green circular area. Besides, it should be noted that the $o_{i_j}$ is obtained by the partial observation $o_i$, therefore the actual meaningful observation of $o_{i_j}$ is the receptive field intersection of agent $i$ and $j$, as shown in the grey-shaded area. } 
     \label{fig:image_obs}
   \end{figure*}
\textit{Imagined Observation} is obtained by function $\mathcal{F}_{im}(o_i, j)$, which is calculated based on the central agent's observation $o_i$.
It represents the relative observation information in the central agent's view, by switching the viewpoints from $i$ to $j$, i.e., think of agent $j$'s coordinates as (0, 0).
Figure \ref{fig:image_obs}  details a clear calculation example in a small-scale environment.
This calculation manner is also applicable to general MARL environments due to the compatibility of their observation feature vector settings.
We take SMAC as an example, the observation feature vector of a single agent is a concatenation of all agents' observation, where the information for invisible agents is set to 0, as shown in Table \ref{table1} and Eq 
 \ref{app_obs_feat},
Therefore, we can imagine agent $j$'s observation by computing relative information from $j$'s perspective, 
which includes relative distances, relative coordinates, corresponding unit health, unit shield, and unit type.
This computation {manner} also applies to GRF \cite{kurach2020google} and MPE \cite{mordatch2018emergence} because they have similar observation feature structures.


\section{Definition and Proofs}
\subsection{Definition of  RA-CTDE}
\begin{definition}
\label{definition}
    (Reward-Additive CTDE).  Let $\bm{\theta} =\{\theta_i\}_{i=1}^{{N}}$ be the parameters of $Q$ functions, $\mathcal{F}$ be the mixing network in CTDE, $\mathcal{N}\!=\!\{1,...,N\}$ be the agents set, $\mathcal{Q}^{N}\!\!=\!\{Q_1(\tau_1,a_1; \theta_1), Q_2(\tau_2,a_2; \theta_2),..., Q_{{N}}(\tau_{{N}},a_{{N}}; \theta_{{N}})\}, \{\bm{\tau},\bm{a},\\{r^{ext}}, \bm{\tau^{\prime}}\} \!\sim\! \mathcal{D}$, assume $\forall i,j \in \mathcal{N}, \theta_i \neq \theta_j$, then Reward-Additive CTDE means computing $\{\mathcal{L}_i^{E}(\theta_i, \phi)\}_{i=1}^{N}$ in Eq \ref{app_definition1} and Eq \ref{app_definition2}  and then updating their parameters respectively. The term $\mathcal{P}$ is not involved in the gradient calculation but as a scalar. Formally{:}
\begin{align}
\mathcal{L}_i^{E}(\theta_i, \phi)
  =\mathbb{E}_{\bm{\tau},\bm{a},{r^{ext}}, \bm{\tau^{\prime}}\in \mathcal{D}}
    \left[ \mathcal{P}
\cdot 
\mathcal{F}(\mathcal{Q}^{{N}},s ; \phi)
\right] 
\label{app_definition1} \\
\mathcal{P}=-2\left(r^{ext}+\gamma \max_{\bm{a}'} {Q_{T}(\bm{\tau',a'}}) -\mathcal{F}(\mathcal{Q}^{{N}},s ; \phi)\right)
 \label{app_definition2}
\end{align} 
\end{definition}

\subsection{Proof of Theorem 1}
\begin{theorem}
Let $\{\theta_i\}_{i=1}^{{N}}$ be the parameters of $Q$ functions, $\phi$ be the parameters of the mixing network $\mathcal{F}$ in CTDE, $\mathcal{L}^G$ be the global target in Eq \ref{app_CTDE_LOSS}, $\mathcal{N}=\{1,...,N\}$ be the agents set, $\mathcal{Q}^{N}\kern-9pt=\kern-9pt\{Q_1(\tau_1,a_1; \theta_1), Q_2(\tau_2,a_2; \theta_2),..., Q_{{N}}(\tau_{{N}},a_{{N}}; \theta_{{N}})\}$,  $\{\bm{\tau},\bm{a},{r^{ext}}, \bm{\tau^{\prime}}\} \kern-2pt\sim\kern-2pt \mathcal{D}$, assume $\forall i,j \in \mathcal{N}, \theta_i \neq \theta_j$, then $\forall i \in \mathcal{N}$ , the following equations hold{:}
\begin{align}
\frac{\partial \mathcal{L}^{G}(\bm{\theta}, \phi)}{\partial \theta_i} &= \frac{\partial \mathcal{L}_i^{E}({\theta_i, \phi})}{\partial \theta_i}
\label{app_theorema} \\
\frac{\partial \mathcal{L}^{G}(\bm{\theta}, \phi)}{\partial \phi} = &\frac{1}{N} \cdot \sum_{i=1}^{N}\frac{\partial \mathcal{L}_i^{E}({\theta_i, \phi})}{\partial \phi}
 \label{app_theoremb}
\end{align} 
\label{app_theorem}
\end{theorem}

\begin{proof}[ Proof]
In the Centralized Training with Decentralized Execution (CTDE) paradigm, $Q$ functions are trained by the following global TD-error target:
\begin{align}
\mathcal{L}^{G}\!(\bm{\theta}\!, \phi) \! = \!\! \mathbb{E}\left[r^{ext} \!\!+\! \gamma \!\max_{\bm{a}'} {Q_{T}\!(\bm{\tau'\!,a'}})\! -\!
\mathcal{F}(\mathcal{Q}^{N},s)
\! \right ]^2 
\label{app_CTDE_LOSS}
\end{align}

We use the chain rule and we get the following derivations:
\begin{align}
\frac{\partial \mathcal{L}^G(\bm{\theta}, \phi)}{\partial \theta_i} \!&\!=\mathbb{E} \left[ \frac{\partial \mathcal{L}^G}{\partial \mathcal{P} }
\cdot
\frac{\partial  \mathcal{P}}{\partial \mathcal{F}(\mathcal{Q}^{{N}},s)}
\cdot
\frac{\partial \mathcal{F}(\mathcal{Q}^{{N}},s; \bm{\theta}, \phi))}{\partial Q_i}
\cdot
\frac{\partial Q_i}{\partial \theta_i}
\right] \nonumber\\
&=\mathbb{E} \left[ -\mathcal{P}
\cdot
(-1)
\cdot
\frac{\partial \mathcal{F}(\mathcal{Q}^{{N}},s; \bm{\theta}, \phi)}{\partial Q_i}
\cdot
\frac{\partial Q_i}{\partial \theta_i}
\right]  \nonumber\\
&=\mathbb{E} \left[ \mathcal{P}
\cdot
\frac{\partial \mathcal{F}(\mathcal{Q}^{{N}},s; \bm{\theta}, \phi)}{\partial Q_i}
\cdot
\frac{\partial Q_i}{\partial \theta_i}
\right] 
\label{theorem_proof111}
\end{align}

\begin{align}
\frac{\partial \mathcal{L}^G(\bm{\theta}, \phi)}{\partial \phi} \!&\!=\mathbb{E} \left[ \frac{\partial \mathcal{L}^G}{\partial \mathcal{P} }
\cdot
\frac{\partial  \mathcal{P}}{\partial \mathcal{F}(\mathcal{Q}^{{N}},s)}
\cdot
\frac{\partial \mathcal{F}(\mathcal{Q}^{{N}},s; \bm{\theta}, \phi))}{\partial \phi} \right] \nonumber\\
&=\mathbb{E} \left[ -\mathcal{P}
\cdot
(-1)
\cdot
\frac{\partial \mathcal{F}(\mathcal{Q}^{{N}},s; \bm{\theta}, \phi)}{\partial \phi}\right]  \nonumber\\
&=\mathbb{E} \left[ \mathcal{P}
\cdot
\frac{\partial \mathcal{F}(\mathcal{Q}^{{N}},s; \bm{\theta}, \phi)}{\partial \phi}
\right] 
\label{theorem_proof2}
\end{align}

Based on the Definition \ref{definition}, we can easily deduce the following derivations:
\begin{align}
\frac{\partial \mathcal{L}^E_i(\theta_i, \phi)}{\partial \theta_i} \!\!
=\mathbb{E} \left[ \mathcal{P}
\cdot
\frac{\partial \mathcal{F}(\mathcal{Q}^{{N}},s; \theta_i, \phi)}{\partial Q_i}
\cdot
\frac{\partial Q_i}{\partial \theta_i}
\right] 
\label{theorem_proof3}
\end{align}

\begin{align}
\frac{\partial \mathcal{L}^E_i(\bm{\theta}, \phi)}{\partial \phi} \!
&=\mathbb{E} \left[ \mathcal{P}
\cdot
\frac{\partial \mathcal{F}(\mathcal{Q}^{{N}},s; \bm{\theta}, \phi)}{\partial \phi}
\right] 
\label{theorem_proof4}
\end{align}

Since the derivatives of $\phi$ computed for each target $\mathcal{L}^E_i({\theta}_i, \phi)$ are same, we obtain the following conclusions:

\begin{align}
\frac{\partial \mathcal{L}^{G}(\bm{\theta}, \phi)}{\partial \theta_i} &= \frac{\partial \mathcal{L}_i^{E}({\theta_i, \phi})}{\partial \theta_i}
\label{theorem_proof5} \\
\frac{\partial \mathcal{L}^{G}(\bm{\theta}, \phi)}{\partial \phi} = &\frac{1}{N} \cdot \sum_{i=1}^{N}\frac{\partial \mathcal{L}_i^{E}({\theta_i, \phi})}{\partial \phi}
 \label{theorem_proof6}
\end{align} 

\end{proof}

\subsection{Proof of Corollary}
\label{combing_way}
\begin{corollary}
Let $\{\theta_i\}_{i=1}^{{N}}$ be the parameters of $Q$ functions,  $\mathcal{N}\!=\!\{1,...,N\}$ be the agents set,  $\mathcal{L}_i^{VDN}$ be the $\mathcal{L}_i^{E}$ 's special case of VDN, $\{\bm{\tau},\bm{a},{r^{ext}}, \bm{\tau^{\prime}}\} \!\!\sim\!\! \mathcal{D}$, assume $\forall i,j \in \mathcal{N}, \theta_i \neq \theta_j$,  
 then $\forall i \in \mathcal{N} ${:} 
\begin{align}
 &\mathcal{L}_{i}^{VDN}(\bm{\theta})=\mathbb{E}_{\bm{\tau},\bm{a},{r^{ext}}, \bm{\tau^{\prime}}\in \mathcal{D}}\left[ \mathcal{P}_i \cdot Q_i({\tau_i,a_i}; \theta_i)\right]
\label{app_eq_corollary1} \\
&\mathcal{P}_i\!\!=\!\!-2\!\left(\!r^{ext}\!\!+\!  \mathcal{R}_i^{VDN} \!\!\!+\! \gamma\! \max_{a'} {Q^{-}_i\!({\tau_i',a_i'}}) \!\!-\! Q_i({\tau_i,a_i})\! \right)
\label{app_eq_corollary2}\\
&\mathcal{R}_i^{VDN} \!=\! \gamma \!\max_{a'}\!\!\sum_{j=1,j \neq i}^{{N}} \!\!{Q^{-}_j({\tau_j', a_j'}}) \!-\!\!\!\! \sum_{j=1, j\neq i}^{{N}}\!Q_j({\tau_j, a_j})
\label{app_eq_corollary3}
\end{align}

\label{app_corollary}
\end{corollary}

\begin{proof}[ Proof]
VDN training paradigm uses linear value factorization, specifically,
\begin{align}
Q_T(\bm{\tau^{\prime}}, \bm{a^{\prime}}) &= \sum_{i=1}^{N} Q_i^{-}(\tau_i^{\prime}, a_{i}^{\prime}) \\
\mathcal{F}(\mathcal{Q}^{N},s) &= \sum_{i=1}^{N}Q_i(\tau_i, a_i)
\end{align}

We define $\mathcal{L}_{i}^{VDN}(\theta_i)$ as in Eq \ref{app_eq_corollary1} then we can derive that:
\begin{align}
\frac{\partial \mathcal{L}^{G}({\bm{\theta}}) }{\partial \theta_i} &=
\mathbb{E} \left[ \mathcal{P}_i
\cdot
\frac{\partial \mathcal{F}(\mathcal{Q}^{{N}},s; \bm{\theta}, \phi)}{\partial Q_i}
\cdot
\frac{\partial Q_i}{\partial \theta_i}
\right] \nonumber \\
&=\mathbb{E} \left[ \mathcal{P}_i
\cdot
1
\cdot
\frac{\partial Q_i}{\partial \theta_i}
\right] \nonumber \\
&=\frac{\partial \mathcal{L}_{i}^{VDN}({\theta_i}) }{\partial \theta_i}
\label{app_eq_corollary2_VDN}
\end{align}
where:
\begin{align}
&\mathcal{P}_i\!\!=\!\!-2\!\left(\!r^{ext}\!\!+\!  \mathcal{R}_i^{VDN} \!\!\!+\! \gamma\! \max_{a'} {Q^{-}_i\!({\tau_i',a_i'}}) \!\!-\! Q_i({\tau_i,a_i})\! \right) \\
&\mathcal{R}_i^{VDN} \!=\! \gamma \!\max_{a'}\!\!\sum_{j=1,j \neq i}^{{N}} \!\!{Q^{-}_j({\tau_j', a_j'}}) \!-\!\!\!\! \sum_{j=1, j\neq i}^{{N}}\!Q_j({\tau_j, a_j})
\end{align}
\end{proof}

Now we have proven the Corollary \ref{app_corollary}. It's worth noting that the loss function $ \mathcal{L}_i^{IQL}$ of IQL is:
\begin{align}
    \mathcal{L}_i^{IQL}(\theta_i) &=\mathbb{E}\left[ \mathcal{P}_i^{\prime}
\cdot
Q_i({\tau_i,a_i}) \right] \nonumber\\
\mathcal{P}_i^{\prime} &= \left( 
r^{ext}\!\!+\! \gamma\! \max_{a'} {Q^{-}_i\!({\tau_i',a_i'}}) \!\!-\! Q_i({\tau_i,a_i})\right)
\label{IQL_training}
\end{align}

By comparing the difference between Eq \ref{eq_corollary3} and Eq \ref{IQL_training}, we can conclude that the VDN is a special type of IQL, except that it adds an intrinsic reward $\mathcal{R}_i^{VDN}$.
In other words, the superior performance of VDN over IQL can be attributed to this reward function $\mathcal{R}_i^{VDN}$.
Besides, it should be noticed that: 
1) $\mathcal{R}_i^{VDN}$ is calculated through the interactive information from other agents, i.e. the difference between the target $Q$ functions and $Q$ functions. 
2) The intrinsic rewards $\{\mathcal{R}_i^{VDN}\}_{i=1}^{N}$ possessed by each agent are different, which allows each agent to have its own reward guidance rather than a uniform one.
Similarly, 
our reward design and integration way incorporate these two aspects, which also consider other agents' information and provide each agent with an individual guidance signal.
This demonstrates the rationality of our IAM method.

\section{Discussions}

\paragraph{Discussion: The IAM impact on credit assignment} 
CTDE essentially completes the potential credit assignment for each agent, which is entirely determined by extrinsic team reward. 
The specific assignment process can be illustrated in Eq \ref{theorem_proof111}. When the mixing network is fixed, the larger the absolute value of $\mathcal{P}$ is, the more the allocation agent $i$ can gain.
However, when team rewards are sparse, the calculated credit assignment may not be entirely reasonable.
By adding intrinsic rewards into the first term, we can alleviate the insufficient credit assignment caused by sparse rewards.

\paragraph{Discussion: Reward-adding way} 
Please note that our reward integrating way in Eq \ref{eq:app_IAM_loss1}, \ref{eq:app_IAM_loss2}  is different from that in EMC \cite{zheng2021episodic}. The latter calculates $N$ individual intrinsic rewards for each agent but averages them to one then adds it to the global team reward, i.e. $r^{int} = \frac{1}{{N}}\sum_{i=1}^{{N}} || \tilde{Q}_i (\tau_i,\cdot) - Q_i(\tau_i, \cdot)||_2 $. 
This shared intrinsic reward doesn't  reward  agents differently, i.e. $\forall i,j \in \mathcal{N}, r^{int}_i=r^{int}_j$ in Eq  \ref{eq:app_IAM_loss2}, 
which provides an overall curiosity evaluation of all agents' behaviors.  
In contrast, our reward-adding way retains the individual behavior difference between agents and rewards good and bad actions differently as $\mathcal{R}^{VDN}_i$ does. 
This allows intrinsic rewards to provide more tailored and personalized feedback to each agent, which helps to improve the overall performance.

\section{Experiment Settings}
\label{Experiment and Implementation Details}

In this section, we describe the physical meanings of observations in three environments, which help illustrate the feasibility of computing \textit{imagine observations} by switching viewpoints. 

\subsection{SMAC}
The StarCraft Multi-Agent Challenge is a well-known real-time strategy (RTS) game.
Akin to most RTSs, independent agents must form groups and work together to solve challenging combat scenarios, facing off against an opposing army controlled centrally by the game's built-in scripted AI.
The controlled units only access local observations within a limited field of view.  
In the agent's sight range, the feature vector contains attributes related to both allied and enemy units, including \textit{distance from the unit, relative x, and y positions, health and shield status,} as well as \textit{the type of unit}, i.e. what kind of unit it is.
Details are shown in Table \ref{table1}. The observation feature of each agent can be concatenated by different types of features, as shown in Eq. \ref{app_obs_feat}. 
The abbreviations are detailed in Table \ref{table1}. 
It is worth noting that although each agent can only observe local information within its field, 
the information of invisible agents is still recorded in the corresponding vector position, where it is set to 0.

\begin{table*}[htb!]
\centering
\begin{tabular}{cccc}
\toprule
Feature&Physical Meaning& \\
\midrule
Agent movement features (amf)&where it can move to, height information, pathing grid \\
Enemy features (ef)& available to attack, health, relative x, relative y, shield, unit type \\
Ally features (af)& visible, distance, relative x, relative y, shield, unit type\\
Agent unit features (auf)&health, shield, unit type\\
\bottomrule
\end{tabular} 
\caption{\textbf{Feature vectors of observations in SMAC}}
\label{table1}
\end{table*}

\begin{table*}[htb!]
  \centering
  \begin{tabular}{cccc}
  \toprule
  Feature&Physical Meaning& \\
  \midrule
  Ball information& position,  direction,  rotation, ball-owned team, ball-owned player \\
  Left team&   position and direction,   tired factor,  yellow card,   active,  roles \\
  Right team&   position and direction,   tired factor,  yellow card,   active,  roles \\
  Controlled player information& active, designated, sticky actions\\
  Match state&score, steps left, game mode\\
  \bottomrule
  \end{tabular} 
  \caption{\textbf{Feature vectors of observations in GRF}}
  \label{table_grf}
\end{table*}


\begin{equation}
\begin{aligned}
Observation\_&Feature = Concat( amf, \\ 
 & ef \times Enemy\_number, \\
 & af \times Ally\_number, \\
 & auf )
\end{aligned}
\label{app_obs_feat}
\end{equation}

The action space of each agent includes: \textit{do nothing, move up, move down, move left, move right, attack}\_[\textit{enemy ID}].
At each time step, every agent takes actions and then the environment returns a global team reward, which is computed by the accumulative damage point and the killing bonus dealt to the enemies.

\subsection{Multi-agent Particle Environment}
As shown in Figure \ref{fig:mpe_example}, the \textit{Cooperative Navigation} is a task that tests the quality of team cooperation, where agents must cooperate to reach a set of $L$ landmarks.
Besides, Agents are collectively rewarded based on the occupancy of any agent on any goal location and penalized when colliding with each other. 
At every time step, each agent obtains the features of the neighbors and landmarks, which are illustrated in Table \ref{table_mpe}, and then takes action from the action space (\textit{move up, move down, move left, move right, do nothing}).
After that, the environment gives a global team reward as feedback, which is based on the occupancy of goal locations and the collision with each other.

\begin{figure}[htbp]
  \centering
  \includegraphics[width=0.25\textwidth]{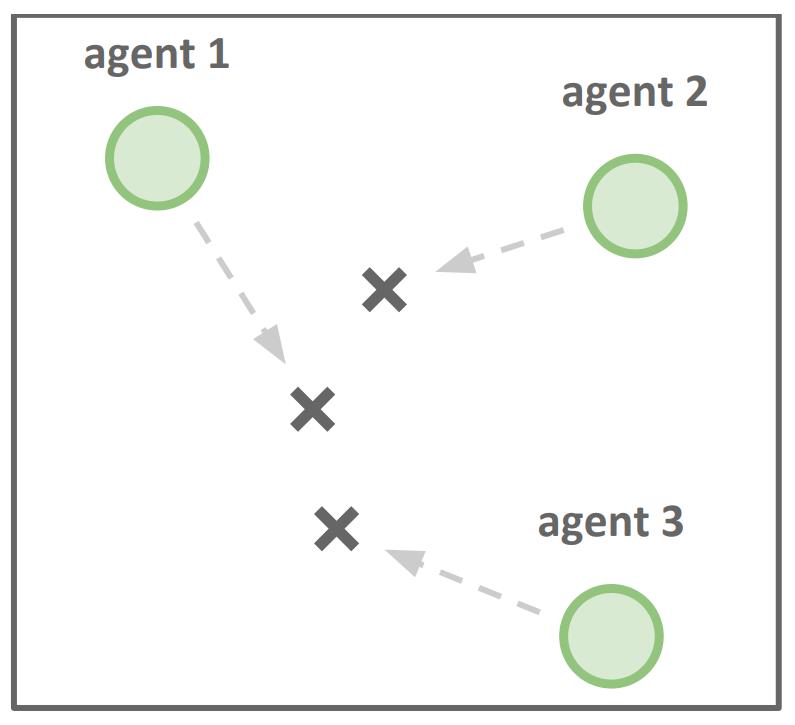}
  \caption{\textit{Cooperative Navigation}}
  \label{fig:mpe_example}
\end{figure}

\begin{table}[H]
\centering
\begin{tabular}{cccc}
\toprule
Feature&Physical Meaning& \\
\midrule
Agent features& visible mask,  position, velocity\\
Landmark features& visible mask, position, velocity\\
\bottomrule
\end{tabular} 
\caption{\textbf{Feature vectors of observations in MPE}}
\label{table_mpe}
\end{table}

\subsection{Google Research Football}
\textit{Football Academy} is a subset of the Google Research Football (GRF) \cite{kurach2020google} that contains diverse small-scale training scenarios with varying difficulty. 
The three challenging scenarios tested in our experiments are described as below.
The observation in \textit{Football Academy} is global to every agent, which differs in agent ID. 
For the sake of convenience in observation calculation, we use the \textit{Floats} (115-dimensional vector) observation setting in GRF \cite{kurach2020google}, which consists of players' coordinates, ball possession, ball direction, active player and game mode. 
Details are described in Table \ref{table_grf}. 
The action space consists of \textit{move actions} (in 8 directions), \textit{different ways to kick the ball}(short and long passes, shooting, and high passes), \textit{sprint}, \textit{intercept}.
The environment feedback a sparse global team reward $r$ in the end, i.e. +1 bonus when scoring a goal and -1 bonus when conceding one.

\paragraph{Academy\_run\_pass\_and\_shoot\_with\_keeper.}Two of our players attempt to score from outside the penalty area, with one positioned on the same side as the ball and unmarked, while the other is in the center next to a defender and facing the opposing keeper.

\paragraph{Academy\_pass\_and\_shoot\_with\_keeper.}Two of our players attempt to score from outside the penalty area, with one positioned on the same side as the ball and next to a defender, while the other is in the center unmarked and facing the opposing keeper.

\begin{figure*}[ht!]
  \centering
  \includegraphics[width=0.9\textwidth]{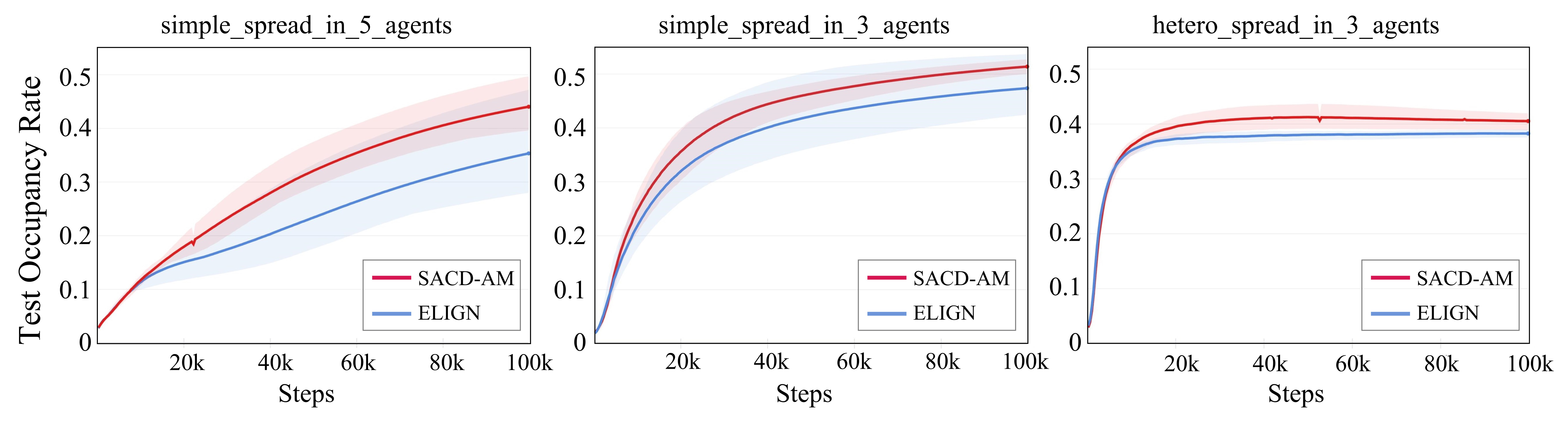}
  \caption{Experiments in different scenarios of \textit{Cooperative Navigation} in MPE.}
  \label{fig:mpe_exp}
\end{figure*}



\section{Additional Experiments}
\paragraph{IAM can improve the performance of VDN.}
Figure \ref{fig:vdn_iam} shows the results of the performance comparisons in SMAC. We integrate IAM with VDN and denote it by VDN-IAM. In order to reflect the impact of IAM on VDN, we utilize the following baselines: (1) We remove the \textit{Episodic Memory} design from EMC \cite{zheng2021episodic} and incorporate it with VDN \cite{sunehag2017value}, which is denoted by VDN-EmC. (2) We use the world model based intrinsic reward in a decentralized way as ELIGN \cite{mguni2021ligs} does and denote it as IQL-WM. (3)LIIR \cite{du2019liir}. (4) VDN \cite{sunehag2017value}. (5) IQL \cite{matignon2012independent}.
When using linear value factorization in CTDE, VDN-IAM outperforms other baselines in six hard maps and gains a significant performance improvement in \textit{8m\_vs\_9m}, \textit{3s5z\_vs\_3s6z}.

\paragraph{Performance comparison with the world model.}
\label{performace_elign}
The \textit{Cooperative Navigation} task establishes barriers to training cooperative behaviors among the agents, as they must collaborate to reach different reaching points to obtain rewards and avoid penalties. 
To assess the performance of the action model intrinsic reward in MPE, we conducted the following experiments: We use the same training settings as ELIGN in 3 scenarios as shown in Figure \ref{fig:mpe_exp}. We employed an independent training approach, combining the action model intrinsic reward and world model intrinsic reward with SACD, denoted as SACD-AM and ELIGN, respectively. 
The test occupancy rate in cooperative navigation reflects the degree of cooperation among agents.
As shown in Figure \ref{fig:mpe_exp}, the action model intrinsic reward exhibited a significant improvement in test occupancy compared to ELIGN, indicating that it can encourage cooperative behavior more effectively than the world model.

\begin{figure*}[htbp!]
  \centering
  \includegraphics[width=0.9\textwidth]{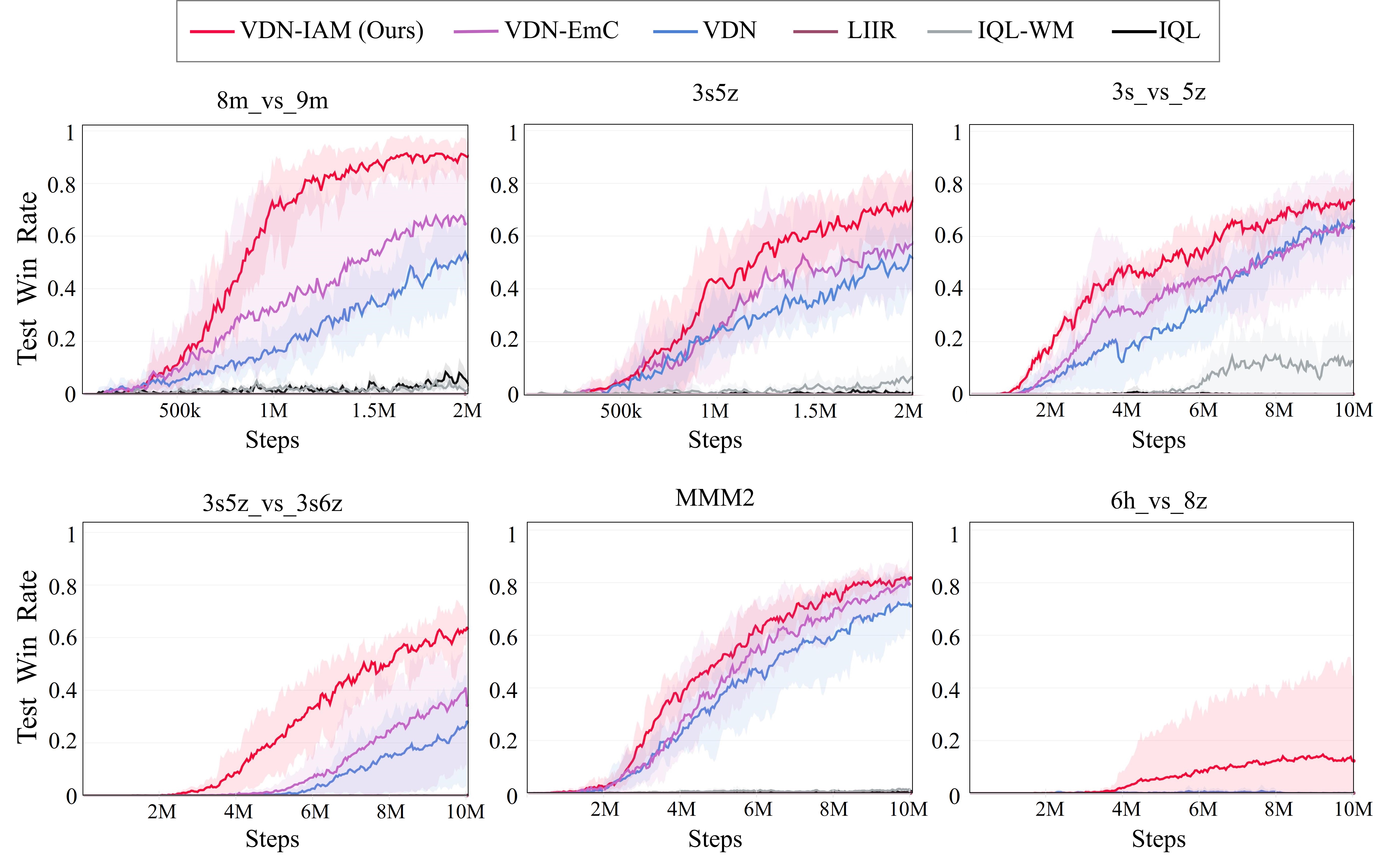}
  \caption{Performance of VDN-IAM in various maps of SMAC.}
  \label{fig:vdn_iam}
\end{figure*}

\begin{figure*}[htbp!]
  \centering
  \includegraphics[width=0.6\textwidth]{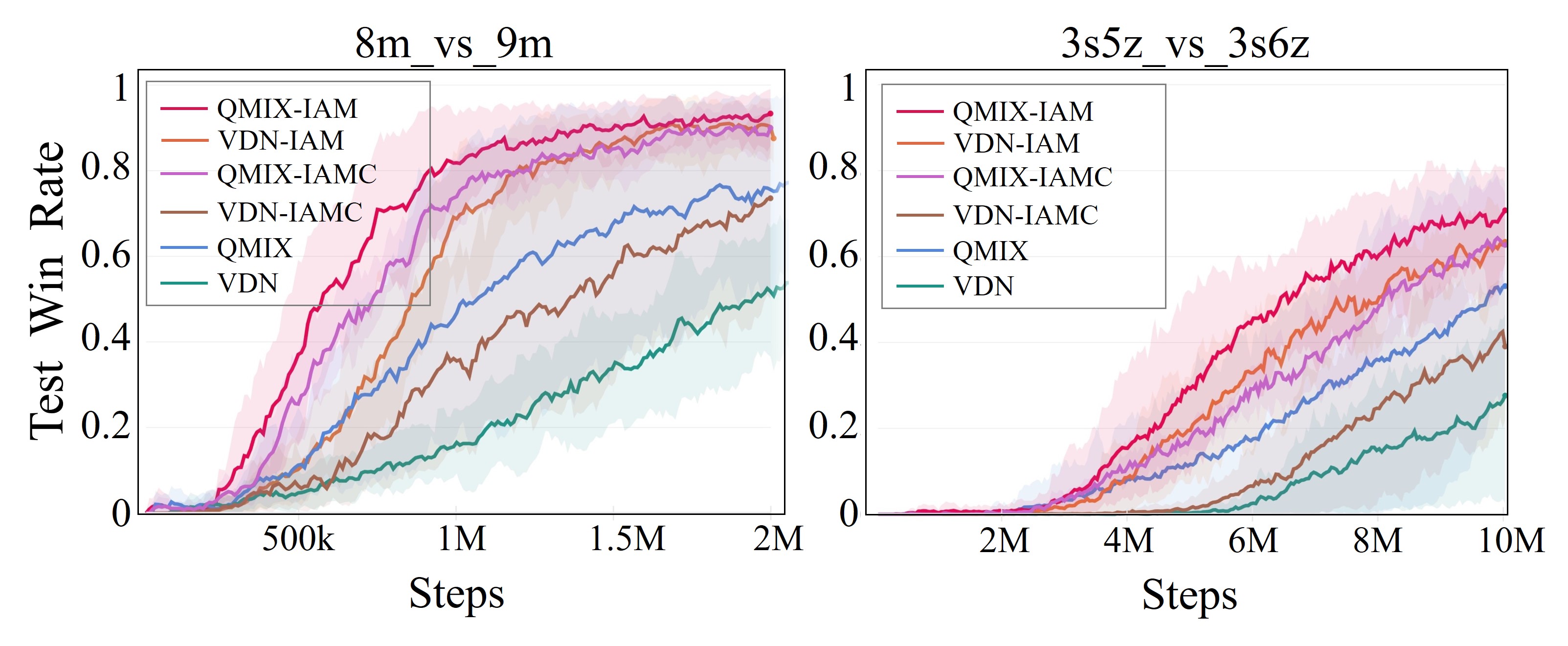}
  \caption{Performance Comparison in different ways of combining IAM-based reward into CTDE.}
  \label{fig:rew_add_4}
\end{figure*}

\begin{figure*}[htbp!]
  \centering
  \includegraphics[width=1\textwidth]{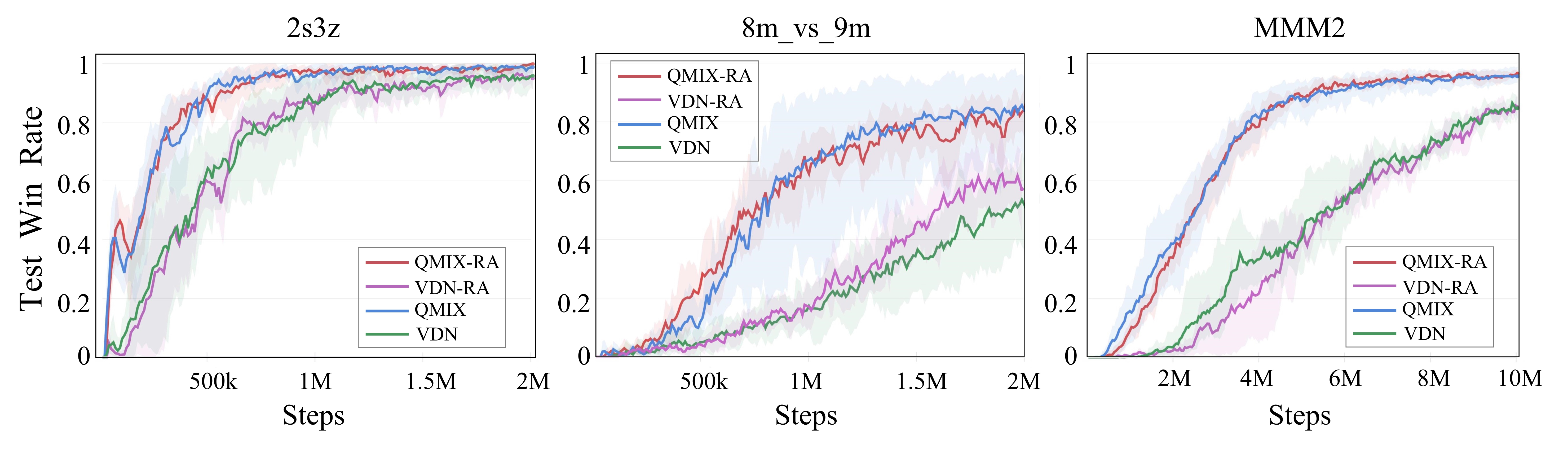}
  \caption{Performance comparisons when no intrinsic rewards are added to RA-CTDE.}
  \label{fig:2_exp}
\end{figure*}

\paragraph{The proposed reward-adding way is better than EMC.}
To demonstrate the different performances in a reward-adding way between IAM and EMC, we conducted the following experiments.
we compare the integration method of RA-CTDE with that of EMC's centralized adding approach. The latter are denoted as VDN-IAMC and QMIX-IAMC.
We test algorithms in \textit{8m\_vs\_9m} and \textit{3s5z\_vs\_3s6z} maps and results can be demonstrated in Figure \ref{fig:rew_add_4}. Although VDN-IAMC and QMIX-IAMC have performance improvement, they are still outperformed by VDN-IAM and QMIX-IAM. This suggests that using RA-CTDE to leverage intrinsic reward is better than using EMC directly.
This combing way is still reasonable from the perspective of credit assignment. 
The intrinsic rewards can distinguish the agents' actions individually, which adds the global TD-loss term in Eq {\ref{eq:IAM_loss2}} during credit assignment and makes team rewards assign more on better actions.

\paragraph{The equivalence of RA-CTDE}
\label{ablation_equ}
To demonstrate the accuracy of Theorem \ref{app_theorem}, we conduct experiments on maps with varying difficulty levels: \textit{2s3z, 8m\_vs\_9m and MMM2}.
Without using additional intrinsic rewards,
we integrate the RA-CTDE training paradigm with QMIX and VDN, denoted as
QMIX-RA and VDN-RA respectively. The baselines are QMIX and VDN.
The results shown in Figure \ref{fig:2_exp} illustrate that the performance of our RA-CTDE training paradigm is equivalent to CTDE when no intrinsic reward is used, which is consistent with our theorem conclusion.

\paragraph{Conclusion} Our experiments and conclusions can be recursively inferred by the following logical reasoning.
1)
Experiments in the paper's explicable example and Figure  \ref{fig:mpe_exp} illustrate that
action model based intrinsic rewards can encourage \textit{action tendency} consistency and outperform the world model based intrinsic reward.
2)
The Theorem \ref{app_theorem} proof and the ablation experiment results in Figure \ref{fig:2_exp} illustrate that RA-CTDE is an equivalent variant of CTDE that can utilize $N$ intrinsic rewards.
3)
The\ textit{Performance Comparison} experiments demonstrate that
the action model intrinsic reward combined with QMIX (denoted by QMIX-IAM) outperforms other baselines.
Since VDN and QMIX represent two main approaches in CTDE (i.e. linear value factorization and nonlinear value factorization), 
the improved performance results 
demonstrate that our reward function can enhance the performance of CTDE to the highest extent.
4)
Experiments in GRF indicate that the
IAM generalizes well in environments with sparse rewards.

\end{document}